\documentclass[11pt,reqno]{amsart}

\usepackage[caption=false]{subfig}
\usepackage[only,rightarrowtriangle,leftarrowtriangle]{stmaryrd}
\usepackage{natbib}
\usepackage[colorlinks,linkcolor=red, citecolor=blue, urlcolor=blue]{hyperref}
\usepackage{algorithm, algpseudocode}
\usepackage{amsthm, amsmath, amssymb}
\usepackage[text={5.8in,8.3in},centering]{geometry}

\theoremstyle{plain}
\newtheorem{theorem}{Theorem}[section]

\newtheorem{corollary}[theorem]{Corollary}
\newtheorem{proposition}[theorem]{Proposition}

\theoremstyle{definition}
\newtheorem{definition}[theorem]{Definition}
\newtheorem{example}[theorem]{Example}

\theoremstyle{remark}
\newtheorem{remark}{Remark}

\makeatletter
\def\lstAZ{A, B, C, D, E, F, G, H, I, J, K, L, M, N, O, P, Q, R, S, T, U, V, W, X, Y, Z}
\def\lstaz{a, b, c, d, e, f, g, h, i, j, k, l, m, n, o, p, q, r, s, t, u, v, w, x, y, z}

\def\lstAZBB{B, C, D, E, F, G, H, I, J, K, L, M, N, O, P, Q, R, T, U, V, W, X, Y, Z}

\newcommand{\MkScr}[1]{\expandafter\def\csname s#1\endcsname{\mathscr{#1}}}
\newcommand{\MkUp}[1]{\expandafter\def\csname u#1\endcsname{\mathrm{#1}}}
\newcommand{\MkBold}[1]{\expandafter\def\csname b#1\endcsname{\mathbf{#1}}}
\newcommand{\MkFrak}[1]{\expandafter\def\csname f#1\endcsname{\mathfrak{#1}}}
\newcommand{\MkCal}[1]{\expandafter\def\csname c#1\endcsname{\mathcal{#1}}}
\newcommand{\MkBB}[1]{\expandafter\def\csname #1#1\endcsname{\mathbb{#1}}}

\@for\i:=\lstAZ\do{%
	\expandafter\MkScr \i  %
	\expandafter\MkFrak \i  %
	\expandafter\MkUp \i %
	\expandafter\MkBold \i %
	\expandafter\MkCal \i  %
}    
\@for\i:=\lstaz\do{%
	\expandafter\MkUp \i   }    

\@for\i:=\lstAZBB\do{%
	\expandafter\MkBB \i     }

\makeatother

\newcommand{\param}[1]{\leftrightarrow  \mathcal{P}\left(#1\right)}
\newcommand{\parambar}[1]{\leftrightarrow  \bar{\mathcal{P}}\left(#1\right)}

\usepackage{tikz-cd}
\usetikzlibrary{circuits.logic.US}
\usetikzlibrary{positioning}
\usetikzlibrary{fit}
\usetikzlibrary{cd}
\usetikzlibrary{arrows}
\usetikzlibrary{calc}
\usetikzlibrary{decorations.markings}
\usetikzlibrary{shapes.geometric}
\tikzset{ed/.style={auto,inner sep=2pt,font=\scriptsize}} %
\tikzset{>=stealth}

\tikzset{vert/.style={draw,circle, minimum size=6mm, inner sep=0pt, fill=white}}
\tikzset{vertblank/.style={ minimum size=6mm, inner sep=0pt, fill=white}}

\tikzset{vertbig/.style={draw,circle, minimum size=8mm, inner sep=0pt, fill=white}}
\tikzset{->-/.style={decoration={
      markings,
      mark=at position #1 with {\arrow{>}}},postaction={decorate}}}

\tikzset{edge/.style={line width=0.5pt, decoration={markings,mark=at position 1 with %
    {\arrow[scale=1.5,>=stealth]{>}}},postaction={decorate}}}

\tikzset{dotted/.style={black!30, line width=0.5pt}}

\pgfdeclarelayer{edgelayer}
\pgfdeclarelayer{nodelayer}
\pgfsetlayers{edgelayer,nodelayer,main}

\tikzstyle{morphism}=[fill=white, draw=black, shape=rectangle]
\tikzstyle{medium box}=[fill=white, draw=black, shape=rectangle, minimum width=0.8cm, minimum height=0.9cm]
\tikzstyle{large morphism}=[fill=white, draw=black, shape=rectangle, minimum width=1.7cm, minimum height=1cm]
\tikzstyle{bn}=[fill=black, draw=black, shape=circle, inner sep=1.5pt]
\tikzstyle{state}=[fill=white, draw=black, regular polygon, regular polygon sides=3, minimum width=0.8cm, shape border rotate=180, inner sep=0pt]
\tikzstyle{medium state}=[fill=white, draw=black, regular polygon, regular polygon sides=3, minimum width=1.3cm, inner sep=0pt, shape border rotate=180]
\tikzstyle{large state}=[fill=white, draw=black, regular polygon, regular polygon sides=3, minimum width=2.2cm, shape border rotate=180, inner sep=0pt]
\tikzstyle{wn}=[fill=white, draw=black, shape=circle, inner sep=1.5pt]

\tikzstyle{arrow}=[->]
\tikzstyle{dashed box}=[-, dashed]

\tikzset{none/.style={%
     append after command={%
       \pgfextra{\node [right] at (\tikzlastnode.mid east) {{\tiny\tikzlastnode}};}
     }}}
\tikzstyle{none}=[]

\newcommand{\dd}{{\,\mathrm d}}

\renewcommand{\th}{\theta}
\newcommand{\la}{\lambda}
\newcommand{\ga}{\gamma}

\newcommand{\eps}{\varepsilon}

\renewcommand{\phi}{\varphi}
\newcommand{\scr}[1]{{\mathcal #1}}

\newcommand{\ind}{\mathbf{1}}

\newcommand{\bs}[1]{{\boldsymbol #1}}
\newcommand{\T}{{\prime}}
\newcommand{\id}{\operatorname{id}}
\providecommand{\trace}{{\operatorname{tr}}}

\newcommand{\E}{\mathbb{E}}

\newcommand{\on}[1]{\operatorname{#1}}

\newcommand{\ch}{\mathrm{ch}}
\newcommand{\pa}{\mathrm{pa}}

\newcommand{\Bm}{\begin{bmatrix}}
\newcommand{\Em}{\end{bmatrix}}

\newcommand{\weight}[3]{w_{#1,#2}(#3)}
\renewcommand{\P}{\mathbb{P}} %

\tikzset{vert/.style={draw,circle, minimum size=6mm, inner sep=0pt, fill=white}}

\begin{document}

\title[Backward Filtering Forward Guiding]{Automatic Backward Filtering Forward Guiding for Markov processes and graphical models}
\title[Backward Filtering Forward Guiding]{Backward Filtering Forward Guiding}

\author{F.~H. van der Meulen}
\address{Department of Mathematics, Vrije Universiteit Amsterdam, The Netherlands}
\email{f.h.van.der.meulen@vu.nl}

\author{M. Schauer}
\address{Department of Mathematical Sciences, 
Chalmers University of Technology and University of Gothenburg, Sweden}
\email{smoritz@chalmers.se}

\author{S. Sommer}
\address{Department of Computer Science DIKU, University of Copenhagen, Denmark}
\email{sommer@di.ku.dk}

 
\maketitle

\begin{abstract}
We develop a general methodological framework for probabilistic inference in discrete- and continuous-time stochastic processes evolving on directed acyclic graphs (DAGs). The process is observed only at the leaf nodes, and the challenge is to infer its full latent trajectory: a smoothing problem that arises in fields such as phylogenetics, epidemiology, and signal processing. Our approach combines a backward information filtering step, which constructs likelihood-informed potentials from observations, with a forward guiding step, where a tractable process is simulated under a change of measure constructed from these potentials. This Backward Filtering Forward Guiding (BFFG) scheme yields weighted samples from the posterior distribution over latent paths and is amenable to integration with MCMC and particle filtering methods. We demonstrate that BFFG applies to both discrete- and continuous-time models,  enabling probabilistic inference in settings where standard transition densities are intractable or unavailable. Our framework  opens avenues for incorporating structured stochastic dynamics into probabilistic programming. We numerically illustrate our approach for a branching diffusion process on a directed tree. 
\end{abstract}

\section{Introduction}

Probabilistic inference in structured dynamical systems is a fundamental task in machine learning, with applications ranging from time series modeling to phylogenetics. In many problems, such as evolutionary biology, signal processing, or latent variable modeling, data is observed only at the \emph{leaves} of a structure, while the latent dynamics evolve along a branching (tree-like) structure. The goal is to infer the latent process given these noisy or partial observations---a problem known as \emph{smoothing}.

In this work, we start from a  smoothing problem on a directed tree with root vertex $r$ and leaf vertices represented by open circles (see Figure~\ref{fig:tree_dag_withkernels}). Each branch either evolves according to a continuous-time stochastic process---such as a diffusion process or continuous-time Markov chain---or follows a transition given by a Markov kernel. At branching vertices, such as $0$ and $3$, the process splits and evolves conditionally independently, given its value at the vertex. Observations are made at the leaf vertices via emission maps that transform the latent state at the parent into an observed value; for example, vertex $5$ might observe a noisy version of vertex $4$. The root state $x_r$ is assumed known, and the branch from $r$ to $0$ encodes a prior.

This setting generalizes classical \emph{state-space models}, where transitions occur in discrete time, and has received attention in \emph{phylogenetics}, where continuous-time models like Brownian motion or Ornstein--Uhlenbeck processes are tractable and well studied (e.g., \cite{hassler2023data}, \cite{ronquist2004bayesian}, \cite{zhang2021large}). However, for many models of interest, closed-form transition densities or analytical smoothing solutions are not available.

We propose a general framework to address this challenge based on three key components:
\begin{enumerate}
    \item \textit{Backward information filtering}, mapping leaf observations into a family of functions $g_u$ on the tree, typically using a simplification of the stochastic process evolving on the tree.
    \item {\it Constructing a \it guided process} on the tree with distribution obtained by an {\it exponential change of measure} to the law of the unconditional (forward) process, using  the family of functions $g_u$. 
    \item \textit{Stochastic simulation}, generating trajectories from the guided process and computing the {\it weight} of each trajectory.
\end{enumerate}
Step (3) involves simulating a \emph{guided process}---a forward process modified by the exponential change of measure to encourage paths consistent with the observations. This overall strategy, which we call \emph{Backward Filtering Forward Guiding (BFFG)}, yields weighted samples from the smoothing distribution and is compatible with MCMC and particle-based inference methods such as the guided particle filter \cite[Chapter 10]{chopin2020introduction}. 
Our approach allows smoothing in complex continuous-time models on tree structures, even when transition densities are intractable or unknown, significantly broadening the class of models amenable to Bayesian inference.

 For both discrete and continuous cases, we identify the key computations required in applying BFFG. Through multiple examples, we demonstrate the versatility of our approach. Moreover, to further illustrate our approach, we show that the  key computations from  \cite{ju2021sequential} and \cite{stoltz2021bayesian} fit within our framework. In a companion paper \cite{schauer2023compositionality} we use category theory to formally prove that BFFG is compositional, a property which we believe to be highly desirable. In a nutshell, it means that BFFG is unaffected by merging two adjacent edges. 
\medskip

We subsequently extend our framework from trees to general directed acyclic graphs (DAGs), ensuring that the guided process retains the same conditional independence structure as the original forward process.
This case is considerably more difficult as, contrary to a directed tree case, there exists no conditioned process that follows the same dependency structure as the unconditional process. Nevertheless, we enforce this dependency structure by retaining the definition of the guided process on the tree. Our main motivation for this choice is that it enables automatic program transform of the unconditional process to the guided process. This entails that the programme structure of BFFG is the same as that of the forward process. This makes our approach amenable to be incorporated in probabilistic programming languages.  The exact way of carrying out the backward filtering step is however problem specific and not automatic.

The literature for inference in graphical models is vast. Nevertheless, as far as we are aware, few works have considered a unified approach to both discrete- and  continuous-time processes evolving along the edges of the graph.

\begin{figure}
\begin{center}
\begin{tikzpicture}

\tikzstyle{empty}=[fill=white, draw=black, shape=circle,inner sep=1pt, line width=0.7pt]
\tikzstyle{solid}=[fill=black, draw=black, shape=circle,inner sep=1pt,line width=0.7pt]

\begin{pgfonlayer}{nodelayer}
		\node [style=empty,label={$r$},] (r) at (-9, 1) {};
		\node [style=solid,label={$0$},] (0) at (-7, 1) {};
		\node [style=solid,label={$1$},] (1) at (-4, -0.5) {};
		\node [style=solid,label={$3$},] (3) at (-4.5, 2) {};
		\node [style=empty,label={$6$},] (6) at (-3, 1) {};
		\node [style=solid,label={$4$},] (4) at (-3, 3) {};
		\node [style=solid,label={$2$}] (2) at (-2, -0.5) {};
		\node [style=empty,label={$7$}] (7) at (-0, -0.5) {};
		\node [style=empty,label={$5$},] (5) at (-1, 3) {};
\end{pgfonlayer}

\begin{pgfonlayer}{edgelayer}
		\draw [style=edge] (r) to (0);
		\draw [style=edge, color=blue] (0) to (1);
		\draw [style=edge, color=blue] (0) to (3);
		\draw [style=edge] (3) to (6);
		\draw [style=edge, color=blue] (3) to (4);
		\draw [style=edge, color=blue] (1) to (2);
		\draw [style=edge] (2) to (7);
		\draw [style=edge] (4) to (5);
\end{pgfonlayer}

\end{tikzpicture}
\end{center}
\caption{Example of a tree with known root vertex $r$, with observations at vertices $5$, $6$ and $7$. A continuous time stochastic process evolves on the branches $(0, 3)$, $(3, 4)$, $(0, 1)$ and $(1, 2)$ which are coloured blue. \label{fig:tree_dag_withkernels}}
\end{figure}
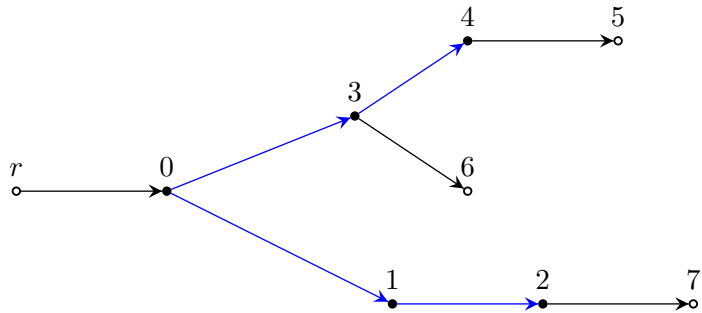

\subsection{Related work}\label{sec:related work}
In case of a line graph with edges to observation leaves at each vertex (a hidden Markov model) there are two well known cases for computing the smoothing distribution: {\it (i)} if $X$ is discrete the forward-backward algorithm (\cite{Murphy2012}, section 17.4.3), {\it (ii)} for linear Gaussian systems the Kalman Smoother, also known as Rauch-Tung-Striebel smoother (\cite{Murphy2012}, section 18.3.2) for the marginal distributions or its sampling version, where samples from the smoothing distribution are obtained by the forward filtering, backward sampling algorithm, \cite{CarterKohn(1994)}.  \cite{Pearl1988} gave an extension of the forward-backward algorithm from chains to trees by an algorithm known as  ``belief propagation'' or ``sum-product message passing", either on trees or poly-trees. This algorithm  consists of two message passing phases. In the ``collect evidence'' phase, messages are sent from leaves to the root; the ``distribute evidence'' phase ensures updating of marginal probabilities or sampling joint probabilities from the root towards the leaves.
The algorithm can be applied to junction trees as well, and furthermore, the approximative loopy belief propagation applies belief propagation to sub-trees of the graph. A review is given in \cite{Jordan2004}. 

\cite{Chou1994} extended the classical Kalman smoothing algorithm for linear Gaussian systems to dyadic trees by using a fine-to-coarse filtering sweep followed by a course-to fine smoothing sweep. This setting arises as a special case of our framework. Extensions of filtering on  Triplet Markov Trees and pairwise Markov trees are dealt with in \cite{BardelDesbouvries2012}  and \cite{DesbouvriesLecomtePieczynski} respectively. 

Particle filters have been employed in related settings, see \cite{doucet2018sequential} for an overview, but the resampling operations of particle filters 
result in the simulated likelihood function being non-differentiable in the parameters, which unlike our approach is an obstacle to gradient based inference. 
\cite{briers2010smoothing} consider particle methods for  state-space models using (an approximation to) the backward-information filter. Twisted particle samplers (\cite{Guarniero2017}) and controlled Sequential Monte Carlo (SMC) (\cite{Heng2020}) essentially use particle methods to find an optimal control policy to approximate the backward-information filter. The approximate backward filtering step in this paper builds on the same idea but extends it to a broader context beyond particle filtering and discrete-time models. A recent application of learning twist functions to large language models is \cite{zhao2404probabilistic}.

For variational inference, \cite{ambrogioni2021automatic} propose an approach which, similar to ours, preserves the Markovian structure of the target  by learning local approximations to the conditional dynamics.
\cite{lindsten2018graphical} consider a  sequential Monte Carlo (SMC) algorithm for general probabilistic graphical models  which can leverage the output from deterministic inference methods.  It shares the idea of using a substitute for optimal twisting functions but does not cover continuoos-time transitions over edges. 
In earlier work (\cite{schauer2017guided}, \cite{mider2021continuous}, \cite{corstanje2023guided}), we introduced guided processes on state-space models for partially observed diffusion processes and chemical reaction networks. These are specific instances of our approach.

\subsection{Outline} 
We start with a short recap of Markov kernels in Section \ref{sec:Markovkernels}. The backward information filter and forward guiding are discussed in sections \ref{sec:BIF} and \ref{sec:guiding} respectively. We then explain the key algorithms of BFFG in Section \ref{sec:implementation}.
Examples that illustrate BFFG are given in sections \ref{sec:examples_discrete} and \ref{sec:examples_continuous} for  discrete and continuous-time procesesses respectively. 
 The extension to a DAG is given in Section \ref{sec:extension_to_DAG}. We illustrate our results in  numerical examples in Section 
\ref{sec:sde_tree}. The appendix contains some technical results and proofs.

Sections \ref{sec:BIF}, \ref{sec:guiding},  \ref{sec:implementation} and  \ref{sec:extension_to_DAG} contain the core ideas of our approach.



\section{Preliminaries: Markov kernels}\label{sec:Markovkernels}
We first recap some elementary definitions on Markov kernels, as these are of key importance in all that follows.

Let $S=(E,\mathfrak{B})$ and $S'=(E',\mathfrak{B}')$ be 
Borel measurable spaces. A Markov kernel between $S$ and $S'$ is denoted by $\kappa\colon S \rightarrowtriangle S'$  (note the special type of arrow), where $S$ is the ``source''  and $S'$ the ``target''. That is, $\kappa\colon E \times \mathfrak{B}' \to [0,1]$, where {\it (i)} for fixed $B\in \mathfrak{B}'$ the map $x\mapsto \kappa(x, B)$ is $\mathfrak{B}$-measurable and {\it (ii)} for fixed $x\in E$, the map $B\mapsto \kappa(x, B)$ is a probability measure on $S'$. 
On a measurable space $S$, denote the sets of bounded measures and bounded measurable functions (equipped with the supremum norm)  by $\bM(S)$ and $\mathbf{B}(S)$ respectively. 
The kernel $\kappa$ induces a {\it pushforward} of the measure $\mu$ on $S$ to $S'$   via 
\begin{equation}\label{eq: pushforward}
\mu \kappa(\cdot)  = \int_{E} \kappa(x, \cdot) \mu(\!\dd x), \qquad \mu \in \bM(S).
\end{equation}
The linear continuous operator $\kappa\colon \mathbf B(S') \rightarrow \mathbf B(S)$ is defined by 
\begin{equation}\label{eq: pullback}
\kappa h(\cdot) = \int_{E} h(y) \kappa(\cdot, \dd y),\qquad  h \in \mathbf{B}(S').
\end{equation} 
We will refer to this operation as the {\it pullback} of $h$ under the kernel $\kappa$. 
Markov kernels $\kappa_1\colon S_0 \rightarrowtriangle S_1 $ and $\kappa_2\colon S_1 \rightarrowtriangle S_2$  
can be composed  by the Chapman-Kolmogorov equations, here written as juxtaposition
\begin{equation}\label{eq:chapman}
(\kappa_1 \, \kappa_2)(x_0, \cdot) = \int_{E_1}  \kappa_2(x_1, \cdot) \kappa_1(x_0, \dd x_1),\qquad x_0 \in E_0.
\end{equation}
The unit $\id$ for this composition is the  identity function considered as a Markov kernel, $\id(x, \dd y) = \delta_{x}(\!\dd y)$.
The product kernel $\kappa \otimes \kappa'$ on $S \otimes S'$
 is defined on the cylinder sets by $
(\kappa  \otimes \kappa') ((x, x'), B \times B') = \kappa (x, B) \kappa' ( x', B'), $  	(where $x\in E$, $x' \in E'$, $B \in \mathfrak{B}$, $B' \in \mathfrak{B}'$) 
and then extended to a kernel on the product measure space.

\section{Backward Information Filter}\label{sec:BIF}
We consider a stochastic process on a tree with vertex set  $\scr{T}$, where the root vertex $r$ is excluded. At each vertex where the process splits, it evolves conditionally independent towards its children vertices. Let $\scr{V}$ denote the set of leaf vertices and  define  $\scr{S} = \scr{T}\setminus \scr{V}$ (the set of non-leaf vertices) Set $\scr{S}_r = \scr{S} \cup \{r\}$.
Thus, in Figure \ref{fig:tree_dag_withkernels}, $\cV=\{5,6,7\}$, $\cS=\{0,1,2,3,4\}$, $\cS_r=\{r, 0,1,2,3,4\}$.  Let $\scr{E}$ denote the set of all edges in the tree. 

 For $t\in \scr{T}$ the state of the process is given by $X_{t}$. For $T \subset \scr{T}$ let $X_T=\{X_t,\, t\in T\}$.  Denote the (unique) parent vertex of vertex $t$ by $\pa(t)$. We assume that the transition to $X_t$, conditional on $X_{\pa(t)}=x$ is captured by a Markov kernel $\kappa_{\pa(t), t}(x, \cdot)$ with source $(E_{\pa(t)}, \scr{B}_{\pa(t)})$ and target $(E_t, \scr{B}_t)$. 

In this section we are interested in the distribution of $X_{\scr{S}}$ conditional on the event  $\{X_{\scr{V}}=x_{\scr{V}}\}$. That is, we are only interested in the values at vertices and not along branches where the continuous time stochastic process evolves. Once we understand this case, we will explain how to adjust for continuous time transitions over an edge.

We assume that there exists a dominating measure $\lambda$ such that for $t\in \scr{V}$
\begin{equation}\label{eq:density_from_leaves} \frac{\dd \kappa_{\pa(t), t}(x,\dd y)}{\lambda_t(\dd y)}(x) =  k_{\pa(t), t}(x,y) .\end{equation}

\subsection{Discrete edges}
We first assume that the probabilistic evolution of the process $X$ is over discrete edges.
\begin{definition}
We call an edge  $e=(s,t)$ {\it discrete}  if it is assumed that on $e$  the probabilistic evolution is governed by the Markov kernel $\kappa_{s,t}$. 
\end{definition}
The process $X$ conditioned on the event $\{X_{\scr{V}}=x_{\scr{V}}\}$ follows the same dependency structure as the unconditional process with transition kernels given by 
\[  \kappa^\star_{\pa(t), t}(x, A) :=\frac{\int_A  \kappa_{\pa(t), t}(x, \dd y) h_t(y)}{\int \kappa_{\pa(t), t}(x, \dd y) h_t(y)} = \frac{(\kappa_{\pa(t), t} h_t \ind_A)(x)}{(\kappa_{\pa(t), t} h_t)(x)}, \qquad t\in \scr{S}. \] Here, if $\scr{V}_t$ denotes the set of leaves descending from vertex $t$, $h_t(x)$ is the density of $X_{\scr{V}_t}$, conditional on $X_t=x$.  In Figure \ref{fig:tree_dag_withkernels} for example, $h_{3}$ is the density of $(X_{5}, X_{6})$ conditional on $X_{3}=x$. The transform of the kernel $\kappa$ to $\kappa^\star$ is known as {\it Doob's $h$-transform}. 

It is well known how the functions $\{h_t,\, t \in \scr{S}\}$   can be computed recursively  starting from the leaves back to the root. These recursive relations have  reappeared in many papers, see for  instance \cite{Felsenstein1981},  \cite{briers2010smoothing}, \cite{Guarniero2017} and \cite{Heng2020}. This recursive computation is known as the {\it Backward Information Filter (BIF)}. 
Firstly, for any leaf vertex  we define
\begin{equation}\label{eq:hleaf} h_{\pa(v) , v} (x)  := 
k_{\pa(v) , v}(x,x_v)   \qquad v \in \scr{V} .\end{equation}
For other vertices $t$ for which $h_t$ has already been computed, set
\begin{equation}\label{eq:hparent} 
h_{\pa(t) , t} :=\kappa_{\pa(t) , t} h_t , \qquad t \in \scr{S}. 
\end{equation}
For a given vertex $t\in \scr{S}_r$, once $h_{t, u}$ has been computed for all children vertices of $t$, which we denote by $\ch(t)$, we get by the Markov property (i.e.\ the process evolves conditionally independent when it branches)
\begin{equation}\label{eq:split_tochilds} h_t(x) = \prod_{u\in \ch(t)} h_{t, u}(x), \qquad t\in \scr{S}_r.
\end{equation}
This can be interpreted as \emph{fusion}, collecting all \emph{messages} $h_{t, u}$  from children at vertex $t$. Summarising, the BIF provides a recursive algorithm that we use to compute $\{h_t,\, t\in \scr{S}\}$. The conditioned process evolves according to the kernels $\kappa^\star$, which can be viewed as applying a change of measure to the kernels $\kappa$.

\subsection{Continuous edges}\label{sec:continuoustime}

Up to this point, we have assumed edges to represent ``discrete time'' Markov-transitions. The probabilistic evolution over a single edge is then captured by a Markov kernel $\kappa$.  Now suppose the process transitions over a continuous edge.
\begin{definition}
We call an edge  $e=(s,t)$ {\it continuous}  if it is assumed that on $e$  a continuous-time process evolves over the time interval $[0, \tau_e]$. The process is assumed to be defined on the filtered probability space $(\Omega, \scr{F}, \mathbb{F}, \PP)$ where  $\mathbb{F}=\{\scr{F}_u,\, u \in [0,\tau_e]\}$. 	The process $X^e:=(X_u,\, u\in [0,\tau_e])$ is assumed to be a  right-continuous, $\mathbb{F}$-adapted Markov process taking values in a metric space $E$.   We identify time $0$ with node $s$ and similarly time $\tau_e$ with node $t$. That is,  $X^e_0=X_s$ and $X^e_{\tau_e}=X_t$. We denote the  infinitesimal generator of the space-time process $((u,X_u),\, u \in [0,\tau_e])$ by  $\cA$ and its domain by $\cD_\cA$.
	\end{definition}
Fix a continuous edge $e$.  
 Assume $X_0=x_s$ and existence of transition kernels  $\kappa_{s,t}$ such that  
 \[
 \PP(X_u\in \dd y \mid X_{u'} = x) = \kappa_{u',u}(x, \dd y), \qquad 0<u'<u<\tau_e.
 \] 
Suppose $h_t$ is given and set $h_{\tau_e}:=h_t$. 
For $u\in [0,\tau_e)$ define the function $(u,x) \mapsto h_u(x)$ as the pullback of $h_{\tau_e}$ under $\kappa_{u, \tau_e}$: 
\begin{equation}\label{eq:pull_u} h_u(x) = (\kappa_{u,\tau_e}h_{\tau_e})(x)= \int h_{\tau_e}(y) \kappa_{u,\tau_e}(x, \dd y). \end{equation}
For convenience, we interchangeably write $h_u(x)$ and $h(u,x)$. For fixed $u$, the map $h_u$ is defined by $x \mapsto h_u(x)$. 
It is well known that  $(u,x) \mapsto h(u,x)$ satisfies the Kolmogorov backward equation  (\cite{bass2011stochastic}, Chapter 37, in particular (37.5))
\begin{equation}\label{eq:kolmbackward}
\cA h_u=	0,\qquad \text{for } u\in [0,\tau_e],\quad  \text{subject to } h_{\tau_e}.
\end{equation} 
Put differently, $h$ is space-time harmonic. 
 This provides the BIF when an edge is governed by a continuous-time process evolving. Finally, to blend this into the description given in the previous section, we have the following definition:
 \begin{definition}
 	For a continuous edge $e=(s,t)$, we define $h_{s,t} = h_0$, where $h_0$ is the solution to \eqref{eq:kolmbackward} at time $0$. 
 \end{definition}
Next, we identify the dynamics of the conditioned process.  Let $\PP_u$ denote the restriction of $\PP$ to $\cF_u$. Define
\[ Z_u = \frac{h_u(X_u)}{h_0(x_s)}. 
\]
By \eqref{eq:pull_u},  $h_u(X_u) = \EE [h_{\tau_e}(X_{\tau_e}) \mid \cF_u]$ is a mean-one martingale.  
For $u\in [0,\tau_e)$ define the  measures $\PP^\star_u$ by 
$\dd \PP^\star_u = Z_{u} \dd \PP_u$ 
and note that $\PP_u^\star$ is the law of the process $(X_u,\, u\in [0,\tau_e])$, restricted to $\scr{F}_u$,  that is conditioned on all observations at leaves descending from vertex $s$. It follows from formula (3.2) in \cite{PalmowskiRolski2002} that the infinitesimal generator $\cA^\star$ of the process under $\PP^\star_t$ can be expressed in terms of $\cL$ and $h$ in the following way
\begin{equation}\label{eq:infgenerator_star} {\cA}^\star f = h^{-1} {\cA}(f h) - h^{-1} f {\cA}  h\end{equation}
(see also \cite{chetrite2015nonequilibrium}, Section 4.1). 
Here, the second term on the right-hand-side must be understood as the composition of three operators:  multiplication by $h^{-1}$, multiplication by $f$ and multiplication by the operator $\cA$  applied to $h$ (the first term on the right-hand-side can be similarly understood). 
Equation \eqref{eq:infgenerator_star} is useful for deriving  the dynamics of the conditioned process.

\begin{example}\label{ex:fig1}
Consider  Figure \ref{fig:tree_dag_withkernels} and observations $\{x_5, x_6, x_7\}$ at leaf vertices. The BIF consists of the following steps:
\begin{itemize}
	\item Initialise from the leaves: set $h_{4,5}(x) = k_{4,5}(x, x_{5})$, $h_{3,6}(x) = k_{3,6}(x, x_{6})$ and $h_{2,7}(x) = k_{2,7}(x, x_{7})$.
	\item Collect incoming messages at vertices $2$ and $4$:  $h_2=h_{2,7}$ and $h_4 = h_{4,5}$.
	\item Solve \eqref{eq:kolmbackward} on edge $(3, 4)$ (with end value $h_4$) to get $h_{3,4}$.\\ Solve \eqref{eq:kolmbackward} on edge $(1, 2)$ (with end value $h_2$) to get $h_{1,2}$.
	\item Collect incoming messages at vertices $1$ and $3$: $h_1=h_{1,2}$ and $h_3=h_{3,4} h_{3,6}$.
	\item Solve \eqref{eq:kolmbackward} on edge $(0, 3)$ (with end value $h_3$) to get $h_{0,3}$.\\ Solve \eqref{eq:kolmbackward} on edge $(0, 1)$ (with end value $h_1$) to get $h_{0,1}$.
	\item Collect incoming messages at vertex $0$: $h_0=h_{0,3}h_{0,1}$. 
\end{itemize}
If a prior measure $\Pi=\kappa_{r,0}$ is specified on the value at vertex $0$. Then $h_r(x_r)= \int h_0(x) \Pi(\dd x)$. We can see that \emph{each edge passes a message to the parent vertex of that edge}. 
	\end{example}

\begin{remark}
 If transitions over edges depend on an unknown parameter $\theta$ then $h_t$ will depend on $\theta$ in general. For notational convenience, we have suppressed this dependency from the notation. 
The likelihood for $\theta$ is, by definition, given by $L(\theta, x_{\cV}):= h_r(x_r)$. 
	\end{remark}

\section{Guiding} \label{sec:guiding}

Example \ref{ex:fig1} shows that the BIF consists of 3 types of operation: {\it (i)} solving Equation \eqref{eq:kolmbackward}, {\it (ii)} pullback under a Markov kernel as defined in \eqref{eq: pullback}, {\it (iii)} pointwise multiplication of functions.  Unfortunately, few cases exist where this can be done in closed form. We propose to resolve this by working with an  approximation $g$ of $h$, for example by choosing a Gaussian approximation to the dynamics of $X$ and computing the $h$-transform for that approximation. Using $g$ rather than $h$ to change the dynamics of the forward process leads to the notion of a guided process. We  show that, contrary to the true conditioned process, the guided  process is tractable. In particular, we can simulate the process. This in turn can be used with stochastic simulation to correct for the  discrepancy induced by using $g$ instead of  $h$.

\subsection{Discrete case}
We start with the definition of the guided process in the discrete case, i.e.\ the case where the process transitions over the vertices $t\in \cT$ and continuous-time transitions over an edge are absent. 

\begin{definition}\label{defn: guided process} 
Assume the maps $x\mapsto g_s(x)$ are specified for each $s\in \cS$.
 We define the \emph{guided process} $X^\circ$   as the process starting in $X^\circ_r = x_r$
 and from the root onwards evolving  \emph{on}  $\scr{S}_r$ according to  transition kernel
\[
\kappa^\circ_{\pa(s), s}(x,\dd y)  =\frac{g_s(y)\kappa_{\pa(s), s}(x,\dd y)}{\int g_s(y)\kappa_{\pa(s), s}(x,\dd y)}
, \qquad s\in \cS.  \] 
\end{definition}
Here, implicitly, we assume that the dominator is strictly positive and finite. 
If $\P^\circ$ denotes the law of $X^\circ_\cS$, then it follows from the definition of $\kappa^\circ$ 
\[  \frac{\dd \P^\circ}{\dd \P}(X_\cS) = \prod_{s\in \cS} \left( \frac{g_s(X_s)}{\int g_s(y)\kappa_{\pa(s), s}(X_{\pa(s)},\dd y)}\right).  \]
\begin{theorem}\label{thm:lr_xstar_xcirc_tree}
\begin{enumerate}
	\item If $\P^\star$ denotes the law of $X_\cS$ conditional on $X_\cV=x_\cV$, then 
\begin{equation}\label{eq:evidence_}
 \frac{\dd \P^\star}{\dd \P}(X_\cS) = \frac{\prod_{v\in \cV} k_{\pa(v), v}(X_{\pa(v)}, x_v)}{h_r(x_r)}. \end{equation}
\item  Assume maps $x \mapsto g_{s, t}(x)$ are specified for each edge $e=(s,t) \in \scr{E}$ and  define $g_s$ by fusion, i.e.\
\begin{equation}\label{eq:split_tochilds_g} g_s(x)= \prod_{t\in \ch(s)} g_{s, t}(x), \qquad s\in \scr{S}_r. \end{equation}
If $\P^\star \ll \P^\circ$, then for $f \in \bB(S_\cS)$
\[\E^\star  f(X_{\cS}) 
 =\frac{g_{r}(x_{r})}{h_r(x_r)} \E^\circ  \left[ f(X_{\cS})  \prod_{s\in \cS \cup \cV} w_{\pa(s), s}(X_{\pa(s)})  \right],
\] with weights defined by  
\begin{equation}\label{eq:treeweights}
w_{\pa(s), s}(x) =
\begin{cases}
\int g_s(y)\kappa_{\pa(s), s}(x,\dd y)   \Big/ g_{\pa(s), s}(x)    & \text{if} \quad s\in \cS\\
h_{\pa(v) , v}(x) \Big/ g_{\pa(v), v}(x) &\text{if} \quad s\in \cV. 
\end{cases}
\end{equation} 
\end{enumerate}
 \end{theorem}
\begin{proof}
By definition of the weights \begin{equation}\label{eq;Pcirc_P}  \frac{\dd \P^\circ}{\dd \P}(X_\cS) = \prod_{s\in \cS} \left( \weight{\pa(s)}{s}{X_{\pa(s)}}^{-1} \times  \frac{g_s(X_s)}{ g_{\pa(s), s}(X_{\pa(s)})}\right).  \end{equation}
By Equation \eqref{eq:split_tochilds_g}, the product of the second term in brackets can be written as 
\begin{equation}\label{eq:cancellation}
	 \frac{ \prod_{s\in\cS}\prod_{t\in \ch(s)} g_{s, t}(X_s) }{ \prod_{s\in\cS} g_{\pa(s), s}(X_{\pa(s)})} = \frac{\prod_{v\in \cV} g_{\pa(v), v}(X_{\pa(v)}) }{g_r(x_r)}
\end{equation} which follows by cancellation of terms (note that the numerator is a product over all edges except those originating from the root node, whereas the denominator is a product over all edges except those ending in a leaf node). Hence 
\begin{equation}\label{eq:Pcirc-P}  \frac{\dd \P^\circ}{\dd \P}(X_\cS) = \frac{\prod_{v\in \cV} g_{\pa(v), v}(X_{\pa(v)}) }{g_r(x_r)} \cdot \left(\prod_{s\in \cS}  w_{\pa(s),s}(X_{\pa(s)})\right)^{-1} . \end{equation} 

The same argument reveals \eqref{eq:evidence_} since the weights  that
 corresponding to $s\in \cS$ equal $1$ when \eqref{eq:dagweights} is evaluated with $h$ rather than $g$. 

The result now follows from 
\[ \frac{\dd \P^\star}{\dd \P^\circ}(X_\cS) =  \frac{\dd \P^\star}{\dd \P}(X_\cS) \cdot \left(  \frac{\dd \P^\circ}{\dd \P}(X_\cS)\right)^{-1} \]
and substituting \eqref{eq:evidence_} and \eqref{eq:Pcirc-P}. 
\end{proof}
Recall that  at vertex $s$, $h_s(x)$ is the likelihood in the subtree with root note $s$ with known value $x$. 
It stands to reason that  good choices for $g_{s}$ approximate $h_s$. 
The definition leaves open the choice of $g_{s, t}$. One way of defining it, is to attach to each kernel $\kappa_{\pa(s),s}$ a kernel $\tilde\kappa_{\pa(s),s}$ (with both kernels having the same source and target) and setting $g_{\pa(s),s}(x)=\int g_s(y)\tilde\kappa_{\pa(s),s}(x,\dd y)$. Then the choice of $\tilde\kappa$ should be such that  calculations in the BIF get simpler. This idea is explored in more detail in concrete examples.
\begin{remark}
	Within the language of Feynman-Kac models (Chapter 5, \cite{chopin2020introduction}), the kernels $\kappa^\circ_{pa(s), s}$ are mutation kernels that can be used in a guided particle filter. Often, these are denoted by $M$ in existing literature. 
\end{remark}

\subsection{Continuous case: exponential change of measure}\label{subsec:continuous}
Suppose $e$ is a continuous edge. The maps $h_u$ for $u\in [0,\tau_e]$ can only be computed in closed form in highly specific settings. For this reason, we propose to replace those maps by substitutes $g_u$. Assume  $(u,x)\mapsto g_u(x)$ is in the domain of $\cA$.
Under weak conditions on $\{g_u,\, u \in [0,\tau_e]\}$ we get that
\[ Z_u^g:=\frac{g_u(X_u)}{g_0(x_s)} \exp\left( -\int_0^u \frac{(\cA g_\tau)(X_\tau)}{g_\tau(X_\tau)} \dd \tau\right) \]
is a mean-one $\cF_u$-local martingale. Corollary \ref{cor:sufficient_localmartingale} in the appendix provides a sufficient conditions for this.

A function $g \colon [0,\tau_e] \times E \to \RR$ for which $Z_u$ is actually a  martingale is called  a \emph{good} function, the terminology being borrowed from \cite{PalmowskiRolski2002}. Sufficient conditions for $g$ to be good  are given in Proposition \ref{prop_goodfunction}. 
Assume we are given a  good function $g$, postpone  the question on how to find such functions in specific settings.
For $u\in [0,\tau_e)$ we can define the  probability measure $\PP_u^\circ$ by 
$\dd \PP_u^\circ = Z_u^g \dd \PP_u$. 
Similar to \eqref{eq:infgenerator_star}, under $\PP^\circ_t$ the process $X$ has infinitesimal generator $\cL^\circ$ satisfying \begin{equation}\label{eq:infgen-guided}
	{\cL^\circ} f = g^{-1} {\cL}(f g) - g^{-1} f {\cL}  g
\end{equation} 
from which the dynamics of the process under $\PP^\circ_u$  can be derived. 
\begin{definition}\label{defn:guidedcontinuous}
The process $X=(X_u,\, u\in [0,\tau_e])$, with $X_0=x_s$, under the law $\PP_{\tau_e}^\circ$ is denoted by $X^\circ$ and referred to as the \emph{guided process induced by $g$ on the edge $e$}.
 \end{definition}

\begin{proposition}\label{prop:continuousguided}
Suppose $X$ evolves on the continuous edge $e=(s,t)$.  If  $X^\circ$ is the guided process induced by $g$ on $e$, then Theorem \ref{thm:lr_xstar_xcirc_tree} remains valid upon defining 
	 \[ g_{s,t}(x) := g_0(x) \qquad \text{and} \qquad  w_{s,t}\left(X^\circ\right)= \exp\left(\int_0^{\tau_e} \frac{{\cA} g}{g} (u, X^\circ_u) \dd u \right).\]
\end{proposition}
Using a tractable approximation $g$ to $h$ we aim to approximate the information brought by future observations, and let this approximation guide the process in a natural way. 
In absence of information from observations, the process evolves just as the unconditional one.

\subsubsection{Choice of $g$}\label{subsec:choice_g}

To forward simulate the guided process, the maps $g_u$ need to be specified along the graph. Ideally, $g_u$ should be like $h_u$ solving $\cA h_u=0$. Typically, the operator $\cA$ can be written as $\partial_u + \cL$, where $\cL$ is the infinitesimal generator of the process $(X_u,\, u\in [0,\tau_e])$, applied to functions where the ``time''-variable $u$ is considered fixed. 

 One option for defining $g_u$ consists of replacing $\cL$  by  $\tilde\cL$, chosen to be  the generator of another continuous-time Markov process. That is, $g$ solves
\begin{equation}\label{eq:kolm_tilde}
	(\tilde\cL + \partial_u) g_u=0,\quad u\in [0,\tau_e]\quad  \text{subject to}\quad  g_{\tau_e}.
\end{equation} 
In this case
$\cA g = (\cL + \partial_u)g= (\cL -\tilde \cL) g$. Sections \ref{ex:sde} and \ref{subsec:ctmc} provide examples of this approach.

Another choice consists of imposing a truncated series expansion of the ansatz  $g(t,x)= \sum_{k=1}^K \alpha_k(t) \psi_k(x)$ and solving 
\[ \partial_t g + P_K (\cL g) =0. \]
Here $P_K f$ projects the function $f$ onto $\{g \colon g(t,x)= \sum_{k=1}^K \alpha_k(t) \psi_k(x)\}$.  In Section \ref{subsec:stoltz} we show how this approach is utilised in \cite{stoltz2021bayesian}.

\section{Backward Filtering Forward Guiding}\label{sec:implementation} 

Application of Theorem \ref{thm:lr_xstar_xcirc_tree} requires to compute $g_t$, $t\in \cS$, from the leaves back to the root vertex, hence traversing the tree backwards. Once this computation is performed, the guided process can be forward simulated according to Definition \ref{defn: guided process}. In simulating the guided process, we traverse the tree once again, in reverse direction. For this reason, we call the joint operation \emph{Backward Filtering Forward Guiding (BFFG)}. Note that the dependency structure of the guided process is inherited from the (unconditional) forward process. 

If we simulate multiple guided processes, then  Theorem \ref{thm:lr_xstar_xcirc_tree} reveals that the relative weight of $X$ is given by
\[g_{r}(x_{r})  \prod_{s\in \cS \cup \cV} w_{\pa(s), s}(X_{\pa(s)}) .\]
Moreover, it forms a positive unbiased estimator for the likelihood $h_r(x_r)$. For multiple algorithms, such as Sequential Monte Carlo (SMC), estimation of the weight by an positive unbiased estimator suffices.

 Denote the collection of $g$-functions that appear in the definition of $\PP^\circ$ by $\scr{G}:=\{g_{s,t},\, e=(s,t)\in \scr{E}\}$. Here, implicitly, for a continuous edge $e$ this includes $\{ g_u,\, u \in [0,\tau_e]\}$. In a companion paper \cite{schauer2023compositionality} we study the compositional structure of BFFG in the language of category theory. While we do not go into these details here, we do want to identify the main requirements.
\begin{enumerate}
	\item $\scr{G}$ can be computed in closed form or a ``good'' choice can be specified right away. 
	\item Sampling under $\PP^\circ$ is tractable. Recall that $\PP^\circ$ is obtained by a change of measure of $\PP$ using $\scr{G}$. It is the law of the guided process.
	\item The weights showing up in either Theorem \ref{thm:lr_xstar_xcirc_tree}  or Proposition \ref{prop:continuousguided} (in case of a continuous edge) can be evaluated. 
\end{enumerate}

For a discrete edge, one generic approach for (1) consists of pairing each kernel $\kappa_{s,t}$ with a kernel $\tilde\kappa_{s,t}$ for which the backward information filter can easily be computed in closed form. In Section \ref{subsec:choice_g} we already listed some strategies for approaching (1) in case of a continuous edge.

\medskip

In sections \ref{sec:examples_discrete} and \ref{sec:examples_continuous} we provide  examples of  discrete- and continuous time guided processes respectively.
In each of these examples  $g$ can be identified from a parameter $\zeta$. In that case  we write  $g\param{\zeta}$. Backward filtering is tractable when
\begin{itemize}
	\item if $g  \param{\zeta}$, then $\tilde \kappa g \param{\zeta'}$ for some $\zeta'$ (this corresponds to the pullback step for a discrete edge);
	\item if $g_{\tau_e}  \param{\zeta}$, then for all $u\in [0,\tau_e]$, the solution to \eqref{eq:kolm_tilde} satisfies $g_u \param{\zeta_u}$ for some $\zeta_u$ (this corresponds to the pullback step for a continuous edge $e$); 
	\item if $g_i  \param{\zeta_i}$, then $\prod_i g_i  \param{\zeta'}$  for some $\zeta'$ (this corresponds to the fusion step).
\end{itemize}
For each of the examples that follow, the results are written in the following form:
\begin{enumerate}
\item {\it Computing $\scr{G}$:}
\begin{itemize}
	\item {\it Pullback for $\tilde\kappa$:}
	\item {\it Fusion:}
	\item {\it Initialisation from leaves:}
\end{itemize}
\item {\it Sampling under $\PP^\circ$:}
\item {\it Computation of the weight:}
\end{enumerate}
Here, in case of a continuous edge $e$, ``pullback for $\tilde\kappa$'' refers to the backward filtering step on the edge $e$ under simplified dynamics of the forward process. 

\section{Examples of guided processes for discrete edges}\label{sec:examples_discrete}

\subsection{Nonlinear Gaussian kernels}\label{subsec:gauss}

Consider the stochastic process on $\cG$ with transitions defined by 
\begin{equation}\label{eq:gaussian_model}X_t \mid X_s = x \sim N(\mu_t(x), Q_t(x)), \qquad t\in \ch(s).  \end{equation}
Unfortunately, for this class of models (parametrised by $(\mu_t, Q_t)$), only in very special cases Doob's $h$-transform is tractable. For that reason,  we look for  a tractable $h$-transform to define a guided process $X^\circ$. This is obtained in the specific case where 
\[\tilde X_t \mid \tilde X_s =x \sim N(\Phi_t x +  \beta_t, Q_t), \qquad t  \in \ch(s). \]
Below we show that backward filtering, sampling from the forward map and computing  the weights are all  fully tractable.

In the following, we write $\phi(x; \mu, \Sigma)$ for the density of the $N(\mu,\Sigma)$-distribution, evaluated at $x$. Similarly, we write $\phi^{\rm{can}}(x; F, H)$ for the density of canonical Normal distribution with potential $F=\Sigma^{-1} \mu$ and precision matrix $H = \Sigma^{-1}$, evaluated at $x$.  

The $g$ functions  can be parametrised by the triple $(c, F, H)$:\begin{equation}\label{eq:gaussian_h} \begin{split} g(y) &= \exp \left(c  + y^\T F -\frac12 y^\T H y \right) \\ & =  \varpi(c,F,H) \phi^{\rm{can}}(y; F, H)= \varpi(c, F, H) \phi(y; H^{-1}F, H^{-1}), \end{split} \end{equation}
where $\log \varpi(c, F, H) = c - \log \phi^{\on{can}}(0, F, H)$.
We write $g\param{c, F, H}$ for $g$ and parameters as in \eqref{eq:gaussian_h}.

In the following result we write $\kappa\equiv \kappa_{\pa(s),s}$, $\tilde\kappa\equiv \tilde\kappa_{\pa(s),s}$ and $w\equiv w_{\pa(s),s}$. 
\begin{theorem}\label{thm:gaussianformulas}
Let $\kappa(x,\dd y) = \phi(y; \mu(x), Q(x)) \dd y$ and $\tilde\kappa(x,\dd y) = \phi(y; \Phi x + \beta, Q)\dd y$.  Assume that $g\param{c, F, H}$ with invertible $H$. 
\begin{enumerate}
\item {\it Computing $\scr{G}$}:
\begin{itemize}
\item Pullback for $\tilde \kappa$: With  $C= Q+ H^{-1}$ invertible,
\[
  \tilde \kappa g  \param{\bar c, \bar F, \bar H}	
  \quad\text{where}\quad
\begin{cases}
	\bar H = \Phi^\T C^{-1} \Phi \\
\bar F = \Phi^\T C^{-1} (H^{-1}F-\beta)\\
	\bar c =  c   - \log \phi^{\rm can}(0, F, H) + \log  \phi(\beta; H^{-1}F, C)
\end{cases}.
\]

\item Fusion: if $g_i \param{c_i, F_i, h_i}$, $i = 1, \dots, k$, then
\[  \prod_{i=1}^k g_i  \param{\sum_{i=1}^k c_i, \sum_{i=1}^k F_i, \sum_{i=1}^kH_i}. \]

\item Initialisation from leaves:  if $v \sim N(\Phi x + \beta, Q)$ is observed at a leaf, then $g_{\pa(v),v}   \param{c, F, H}$ where 
\[
c =   \log  \phi(\beta; v, Q)
 \qquad 
 F = \Phi^\T Q^{-1} (v - \beta) \qquad 
 H = \Phi^\T Q^{-1} \Phi	.
\]
\end{itemize}

\item {\it Sampling under $\PP^\circ$}: \[X^\circ_s \mid X^\circ_{\pa(s)}=x \sim N^{\mathrm{can}}\left( F + Q(x)^{-1} \mu(x), H + Q(x)^{-1} \right).\]

\item {\it Computation of the weight:} if $C(x)= Q(x)+ H^{-1}$ is invertible, then 
$(\kappa g)(x) = \varpi(c, F, H) \phi(H^{-1}F; \mu(x), C(x))$ and  $w(x)$ is obtained from
\[ w(x) = \frac{(\kappa g)(x)}{(\tilde\kappa g)(x)}.\]   
\end{enumerate}
\end{theorem}
\begin{proof}
The proof is elementary.  For the pullback, similar results have appeared in the literature,  for example  Chapter 7 in \cite{chopin2020introduction},   \cite{wilkinson2002conditional} and Chapter 5 in \cite{cappe2005springer}.

First note that 
\begin{align*} (\kappa g)(x) &= \int g(y) \kappa(x,\dd y) =\int \varpi(c,F,H)\phi(y; H^{-1}F, H^{-1}) \phi(y; \mu(x), Q(x)) \dd y
 \\ &= \varpi(c, F, H) \phi(H^{-1}F; \mu(x), Q(x) + H^{-1}).
\end{align*}
By using the specific form of $\mu(x)$ and $Q(x)$ for the kernel $\tilde\kappa$ the expression for the weight follows.

To find the parametrisation of $\tilde\kappa h$, let $C=Q+H^{-1}$. We have
\begin{align*}
(\tilde \kappa g)(x) &= \varpi(c, F, H) \phi(H^{-1}F; \Phi x + \beta, Q + H^{-1}) \\ 
& =\varpi(c, F, H) (2\pi)^{-d/2} |C|^{-1/2} \\ & \quad \times \exp\left(-\frac12 x^\T \Phi^\T C^{-1} \Phi x + (H^{-1}F-\beta)^\T C^{-1} \Phi x - \frac12 (H^{-1}F-\beta)^\T C^{-1} (H^{-1}F-\beta)\right).
\end{align*}
The result follows upon collecting terms. 
If $\Phi$ is invertible, then 
\[ \bar F^\T \bar H \bar F = (H^{-1} F -\beta)^\T C^{-1} (H^{-1} F -\beta). \]
A bit of algebra then gives the stated result. 

The derivation of the parametrisation of the fusion step is trivial. 

 To derive forward simulation under $\PP^\circ$, we write $\propto$ to denote proportionality with respect to $y$
\begin{align*} \frac{g(y) \kappa(x,\dd y)}{\dd y} &\propto \exp\left(-\frac12 y^\T H y + y^\T F\right) \exp\left(-\frac12 (y-\mu(x))^\T Q(x)^{-1} (y-\mu(x))\right) \\ &\propto \exp\left(-\frac12 y^\T (H+ Q(x)^{-1}) y+ y^\T (F+Q(x)^{-1}\mu(x))\right)\\ & \propto \phi^{\mathrm{can}}(y; F + Q(x)^{-1} \mu(x),H + Q(x)^{-1})
\end{align*}
which suffices to be shown.
\end{proof}

Taking an inverse of $H$ respective $C$ can be avoided using
$
(Q+H^{-1})^{-1} H^{-1} = (Q H+I)^{-1} , 
(Q+H^{-1})^{-1}  = H - H (H+Q^{-1})^{-1} H .  
$
While backward filtering for a linear Gaussian process on a tree is well known (see e.g.\ \cite{Chou1994}, section 3), results presented  in the literature often don't state update formulas for the constant $\varpi$ (or equivalently $c$). If the mere goal is smoothing with known dynamics, this suffices. However, we will also be interested in estimating parameters in $\mu$ and or $Q$ then the constant cannot be ignored. In case the dynamics of $X$ itself are linear, then one can take $\tilde\kappa=\kappa$ which implies  $w=1$. Moreover, if additionally the process lives on a ``line graph with attached observation leaves'', where each non-leaf vertex has one child in $\cS$ and at most one child in $\cV$, our procedure is essentially equivalent to Forward Filtering Backward Sampling (FFBS, \cite{CarterKohn(1994)}), the difference being that our procedure applies in time-reversed order. On a general directed tree however the ordering cannot be changed and the filtering steps must be done backwards, as we propose, just like in message passing algoritheorems in general.

\subsection{Discrete state-space Markov chains and particle systems}\label{sec:discrete}

\subsubsection{Branching particle on a tree}
Assume a  ``particle'' takes values in a finite state space $E=\{1,\ldots, R\}$ according to the  $R\times R$ transition matrix $K$.  At each vertex, the particle is allowed to copy itself and branch. 

The algorithmic elements  can easily be identified. The forward kernel $\kappa(x,\dd y)$ can be identified with the  matrix $K$. Furthermore, the map $x\mapsto g(x)$ can be identified with the column vector $\bs{g}=[g_1,\ldots, g_R]^\T$, where $g_i=g(i)$. 
Hence $g$ is parametrised by $\bs{g}$ and we write $g\param{\bs{g}}$.

 Let $\bs{e}_k$ denote the $k$-th standard basis-vector in $\RR^R$ and denote the $j$-th element of the vector $a$ by $\langle a\rangle_j$.
\begin{theorem}\label{thm:finitestate}
Let $\kappa(x,\dd y)$ and  $\tilde\kappa(x,\dd y)$ be represented by the stochastic matrices $K$ and $\tilde K$ respectively.  Assume that $h\param{\bs{g}}$. 
\begin{enumerate}
\item {\it Computing $\scr{G}$:}
\begin{itemize}
	\item {\it Pullback for $\tilde\kappa$:} $\tilde \kappa g  \param{\tilde  K\bs{g}}$
	\item {\it Fusion:} if $g_i  \param{\bs{g}_i}$, $i = 1, \dots, k$, then
\[  \prod_{i=1}^k g_i   \param{\bigcirc_{i=1}^k \bs{g}_i},\] where $\bigcirc$ denotes the Hadamard (entrywise) product.
	\item {\it Initialisation from leaves:} at a leaf vertex $v$ with observation $k\in E$, set $\bs{g}_{\pa(v),v}= \bs{e}_k$.
\end{itemize}
\item {\it Sampling under $\PP^\circ$:}  
\[ \PP^\circ(X_s = k \mid X_{\pa(s)}=\ell) =\frac{\langle \bs{z}_\ell\bigcirc \bs{g}_s\rangle_k}{\bs{z}_\ell^\T \bs{g}_s}, \]
where $\bs{z}_\ell = K_{\pa(s),s}^\T \bs{e}_\ell$ is the $\ell$-th row of $K_{\pa(s),s}$. 
\item {\it Computation of the weight:} The weight at $x\in E$ is given by  \[ w(x) = \frac{(\kappa g)(x)}{(\tilde\kappa g)(x)} =   \frac{\langle K\bs{g}\rangle_x}{\langle\tilde K\bs{g}\rangle_x} \]
\end{enumerate}

\end{theorem}

While we can simply take $\tilde\kappa=\kappa$ on a tree, yielding unit weights, in case $R$ is very large, it can nevertheless be advantageous to  use a different (simpler) map $\tilde\kappa$. To see, this, note that we only need to compute one element of the matrix vector product $K\bs{g}$, but need to compute the full vector $\tilde{K}\bs{g}$ in when backward filtering. Hence, choosing $\tilde K$ sparse can give computational advantages.

\subsubsection{Interacting particles -- cellular automata -- agent based models}

Now consider a discrete time interacting particle process, say with $n$ particles, where each particle takes values in $\{1,\ldots, R\}$. Hence, a particle configuration $x$ at a particular time-instant takes values in $E=\{1,\ldots, R\}^n$.

A simple example consists of particles with state $E = \{1\equiv \mathbf{S}, 2\equiv \mathbf{I}, 3\equiv \mathbf{R}\}$ that represent the health status of individuals that are either {\bf S}usceptible, {\bf I}nfected or {\bf R}ecovered. Then the probability to transition from state {\bf S} to {\bf I} may depend on the number of nearby particles (individuals) that are infected, hence the particles are ``interacting''. A similar example is obtained from time-discretisation of the contact process (cf.\ \cite{Liggett05}).
Note that each particle has multiple parents defined by a local neighbourhood and that a model like this is sometimes referred to as a cellular automaton.

Let $x\in E$. Without additional assumptions, the forward transition kernel 
can be represented by a $R^n \times R^n$ transition matrix. With a large number of particles such a model is not tractable computationally. To turn this into a more tractable form, we make the simplifying assumption that  conditional on $x$ each particle transitions independently. This implies that  the forward evolution kernel $\kappa(x,\dd y)$ can be represented by $(K_1(x),\ldots, K_n(x))$, where $K_i(x)$ is the $R\times R$ transition matrix for the $i$-th particle. Note that the particles interact because the transition kernel for the $i$th particle  may depend on the state of {\it all} particles. 

Due to interactions, it will computationally be very expensive to use $\kappa$ for  backward filtering. Instead, in the backward map we propose to ignore/neglect all interactions between particles. This means that effectively we  backward filter upon assuming all particles move independently, and the evolution of the $i$-th particle  depends only on $x_i$ (not $x$, as in the forward kernel). Put differently, in backward filtering, we simplify the graphical model to $n$ line graphs, one for each particle.

This choice implies that  $\tilde\kappa$ can be represented by $(\tilde K_1,\ldots, \tilde K_n)$ and the map $x\mapsto g(x) = \prod_{i=1}^n g_i(x)$ can be represented by $(\bs{g}^1,\ldots,\bs{g}^n)$. We write $g\param{\bs{g}^1,\ldots,\bs{g}^n}$. The following result is a straightforward consequence of Theorem \ref{thm:finitestate}.

\begin{theorem}
Assume $\kappa$ and $\tilde \kappa$ are represented by $R\times R$ stochastic matrices $K_1,\ldots, K_n$ and $\tilde K_1,\ldots, \tilde K_n$ respectively.  Assume that $g\param{\bs{g}^1,\ldots,\bs{g}^n}$.
\begin{enumerate}
\item {\it Computing $\scr{G}$:}
\begin{itemize}
	\item {\it Pullback for $\tilde\kappa$:} $\tilde \kappa g  \param{\tilde K_1\bs{g}^1,\ldots,\tilde K_n\bs{g}^n}$
	\item {\it Fusion:} if $g_i  \param{\bs{g}^i_1,\ldots, \bs{g}^i_n}$, then
\[ \prod_{i=1}^k g_i  \param{ \bigcirc_{i=1}^k \bs{g}^i_1,\ldots,  \bigcirc_{i=1}^k \bs{g}^i_n}.\]
	\item {\it Initialisation from leaves:} for observation $v\in E$ set $\bs{g}_{\pa(v),v}^i = \bs{e}_{\langle v\rangle_i}$ for $i \in \{1,\ldots, n\}$. 

\end{itemize}
\item {\it Sampling under $\PP^\circ$:} if the forward transition on the edge $(s,t)$ is represented by $(K_1(x),\ldots, K_n(x))$, then with $g_t=\param{\bs{g}^1,\ldots, \bs{g}^n}$
\begin{align*}
	 \PP^\circ(X_t = y \mid X_{s}=x) &= \prod_{i=1}^n  \PP^\circ\left(\langle X_t\rangle_i = \langle  y \rangle_i \mid X_{s}=x\right) \\ &=  \prod_{i=1}^n \frac{\langle \bs{z}^i_\ell\bigcirc \bs{g}^i\rangle_{\langle y\rangle_i}}{(\bs{z}^i_\ell)^\T \bs{g}^i}, 
\end{align*}
where $\bs{z}^i_\ell = K_i(x)^\T \bs{e}_\ell$ is the $\ell$-th row of $K_i(x)$.

\item {\it Computation of the weight:}  Since $(\kappa g)(x) \param{K_1(x)\bs{g}_1,\ldots,K_n(x)\bs{g}_n}$, the weight at $x$ is given by \[ w(x) = \frac{(\kappa g)(x)}{(\tilde \kappa g)(x)} = \prod_{i=1}^n \frac{\langle K_i(x) \bs{g}_i\rangle_{x_i}}{\langle\tilde K_i \bs{g}_i\rangle_{x_i}}. \]

\end{enumerate}
 \end{theorem}

\subsubsection{Backward diagonalisation}\label{sec:backward diagonlisation}
The example of the previous section shows a generic way to deal with interacting particles systems. 
Here, conditional on the state of all particles at a particular ``time'', all particles transition independently, though with transition probabilities that may depend on the state of {\it all} particles. The backward filtering is however done on  separate line graphs, thereby fully bypassing the need of a tractable fusion step. We call this backward diagonalisation.

\subsection{Line graph with independent Gamma increments}\label{subsec:Gamma_example}

In this example, we consider a process with Gamma distributed increments. We consider a line graph  with a single observational leaf $v$. 
We write $Z \sim \on{Gamma}(\alpha,\beta)$ if $Z$ has density $\psi(x; \alpha, \beta) =  \beta ^{\alpha }\Gamma (\alpha )^{-1} x^{\alpha -1}e^{-\beta x} \ind_{ (0,\infty)}(x)$.  

Choose  mappings $\alpha_t$,  $\beta_t\colon (0,\infty) \to [1/C, C]$ for some $C>0$. We define a Markov process  with Gamma increments on $\cG$  by 
\[   	X_{t} - X_s \mid X_s=x \sim  \on{Gamma}(\alpha_t, \beta_t(x)), \quad \text{if $s = \pa(t)$}, \qquad  X_0 = x_0. \] 
This implies
\[ h_{\pa(t),t}(x,y)  = \psi(y - x; \alpha_t, \beta_t(x)). \]
 Note that $X$ can be thought of as a time discretised version of the  SDE $\dd X_t = \beta^{-1}(X_t)\dd L_t$ driven by a Gamma  process  $(L_t)$ with scale parameter 1, observed at final time $T$.  See \cite{Belomestnyetal2019} for a continuous time perspective on this problem. A statistical application will typically involve multiple observations and henceforth multiple line graphs. Simulation on each line graph corresponds to conditional (bridge) simulation for the process $X$.

As the $h$-transform is not tractable, we introduce the process $\tilde{X}$ which is defined as  the Markov process  with Gamma increments  induced by  $(\alpha_t, \tilde\beta)$,  with $\tilde\beta$ a user specified nonnegative constant (which may depend on the vertex $t$; we drop this from the notation). This process is tractable and is used to define $\scr{G}$.

As before, we state our results using the kernels $\kappa$ and $\tilde\kappa$. Define for $A, \beta>0$
\begin{equation}\label{eq:h_gammaexample} g(y) =g(y; A, \beta)= \psi(x_v-y; A, \beta) \end{equation}
and write $g\param{A, \beta}$.
\begin{definition}
 Define the exponentially-tilted Beta-distribution with parameters $\gamma_1, \gamma_2>0$ and $\lambda\in \RR$ as the distribution with density
\begin{equation}\label{eq:exptilted_beta} q_{\gamma_1, \gamma_2, \lambda}(z) \propto z^{\gamma_1-1} (1-z)^{\gamma_2-1} e^{-\lambda z} \ind_{(0,1)}(z). \end{equation}
We denote this distribution by $\mathrm{ExpBeta}(\gamma_1, \gamma_2, \lambda)$.	
\end{definition}
 Sampling from this distribution can be accomplished for example  using rejection sampling, with importance sampling distribution Beta($\gamma_1, \gamma_2$).

\begin{theorem}\label{lem:gamma_incr_htilde_update}
Let $\kappa(x,\dd y) = \psi(y-x; \alpha, \beta(x))\dd y $ and $\tilde\kappa(x,\dd y) = \psi(y-x; \alpha, \tilde\beta)\dd y $.  Assume that  $g\param{A, \beta}$. 

\begin{enumerate}
\item {\it Computing $\scr{G}$:}
\begin{itemize}
	\item {\it Pullback for $\tilde\kappa$:} $\tilde \kappa g \param{A+\alpha, \tilde\beta}$. 
	\item {\it Fusion:} need not be defined as we consider a line-graph.
	\item {\it Initialisation from leaves:} $g_{\pa(v), v)}(x) = \psi(x_v-x; A, \tilde\beta)$. 
\end{itemize}
\item {\it Sampling under $\PP^\circ$:} 
\[ X^\circ_s \mid X_{\pa(s)}^\circ = x \sim x + Z (x_v-x), \]
where $Z\sim \mathrm{ExpBeta}(\alpha, A, \xi(x))$.
\item {\it Computation of the weight:}  the weight at $x$ is given by  \[ w(x) = \frac{(\kappa g)(x)}{(\tilde\kappa g)(x)} =  \left(\frac{\beta(x)}{\tilde\beta}\right)^{\alpha} \EE\, e^{-\xi(x) Z},\]
 where $Z \sim \mathrm{Beta}(\alpha, A)$ and $\xi(x)=(\beta(x)-\tilde\beta)(x_v-x)$.
\end{enumerate}
\end{theorem}

\begin{proof} We have
\begin{align*}
(\kappa g)(x) & = \int_x^{x_v} \psi(x_v-y; A, \tilde\beta) \psi(y-x; \alpha, \beta(x)) \dd y \\ & = \frac{\tilde\beta^A \beta(x)^\alpha}{\Gamma(A) \Gamma(\alpha)} e^{-\tilde\beta x_v +\beta(x) x} \int_x^{x_v} (x_v-y)^{A-1} (y-x)^{\alpha-1} e^{(\beta-\beta(x)) y } \dd y \\ & = 	
\frac{\beta^A \beta(x)^\alpha}{\Gamma(A) \Gamma(\alpha)} e^{-\beta (x_v-x)} (x_v-x)^{A+\alpha-1} \int_0^1 z^{\alpha-1} (1-z)^{A-1} e^{-z\xi(x)} \dd z \\&= \frac{\beta^A \beta(x)^\alpha}{\Gamma(A) \Gamma(\alpha)} e^{-\tilde\beta (x_v-x)} (x_v-x)^{A+\alpha-1} B(\alpha,A) \EE e^{-\xi(x) Z}
\end{align*}
where we made the substitution $y=x+z(x_v-x)$ at the third equality and $\mathrm{B}(\alpha, \beta)$-denotes the Beta-function, evaluated at $(\alpha, \beta)$.
Using the definition of $g$ we obtain 
\[ (\kappa g)(x)= g(x;A+\alpha, \tilde\beta) \left(\frac{\beta(x)}{\tilde\beta}\right) \EE e^{-\xi(x) Z}.
 \]
Upon taking $\xi(x)=\tilde\xi$ we get $\tilde\kappa g = g(\cdot; A+\alpha, \tilde\beta)$. This  also directly gives the expression for the weight. 

Sampling from $\PP^\circ$ over an edge entails drawing from a density which is proportional to $g(y) \kappa(x, \dd y)$ (proportionality with respect to $y$). The claim now follows upon inspecting the derivation of $(\kappa g)(x)$ and keeping all terms under the integral proportional to $z$. 

\end{proof}

Note that the expression for the weight immediately suggests a method to estimate $w(x)$ unbiasedly. 
The setting is restricted to a line graph, as  fusion of $g_1$ and $g_2$ of the form \eqref{eq:h_gammaexample} does not  lead to a fused function of the same form.

The model can be extended to $n$ Markov processes with Gamma increments, which evolve conditionally independent, where it is assumed that the forward evolution consists of composing kernels
 \[ \kappa(\bs{x}, \dd \bs{y})  = \prod_{i=1}^n  \psi(\bs{y}_i - \bs{x}_i; \alpha, \beta(\bs{x})) \dd \bs{y}. \]	
 The backward kernel $\tilde \kappa$ is then taken to be the same, with $\beta(\bs{x})$ replaced with $\beta$. This is an instance of backward diagonalisation.

\subsection{Model by \cite{ju2021sequential}}
\label{subsec:ju}

Consider $N$ particles (agents) taking values in $\{0,1\}$. Let $x \in E:=\{0,1\}^N$. One of the models put forward in \cite{ju2021sequential} is a discrete-time hidden Markov model, where the hidden Markov process has  transition probabilities  given by 
\begin{equation}\label{eq:kappa_ju}
	\kappa(x,y) = \prod_{i=1}^N \alpha_i(x)^{y_i}(1-\alpha_i(x))^{1-y_i}, \qquad  x, y\in E 
\end{equation} 
where 
\begin{equation}\label{eq:kappa_ju2} \alpha_i(x) = \left(\lambda_i a_i(x)\right)^{1-x_i} \left(1-\gamma_i\right)^{x_i},\qquad \lambda_i \in (0,1),\: \gamma_i \in (0,1).\end{equation}
Here, with  $\scr{N}_i$ denoting the indices of the neighbours of particle $i$, 
$\alpha_i(x) = |\scr{N}_i|^{-1} \sum_{m\in \scr{N}_i} \ind\{x_m=1\}$ 
is  the fraction of neighbours of particle $i$ that take value $1$. The parameters $\lambda_i$ and $\gamma_i$ can be further parametrised by a common parameter $\theta$ but  are considered fixed in our discussion. 
The transition probabilities can be summarised by the table
\begin{center}
	\begin{tabular}{c | c c}
		  & $y_i=0$ & $y_i=1$ \\ \hline
		$x_i=0$  & $1- \lambda_i a_i(x)$ & $\lambda_i a_i(x)$  \\
		$x_i=1$  & $\gamma_i$ & $1-\gamma_i$ 
	\end{tabular}
\end{center}

A typical application consists of particles being individuals, with $0$  and $1$ encoding ``susceptible'' and ``infected'' respectively. Furthermore, individual $i$ has its own parameters $\la_i$ and $\ga_i$, which may depend on its  characteristics.  The probability of a susceptible individual becoming infected increases according to the fraction of individuals in its neighbourhood that are infected. 
Let
\[ I(x) = \sum_{i=1}^N \ind\{x_i=1\} \]
denote the total number of infected individuals in the population. 
At time $k$, $k\in \{0,1,\ldots, n\}$, we assume to observe a realisation from the random variable  $V_k \sim \mbox{Bin}(I(x), \rho)$, where $x$ is the state at time $k$ and  the parameter  $\rho$ can be interpreted as a ``reporting probability''.

\medskip
As the dimension of $E$ grows exponentially with $N$, exact backward filtering is not tractable. Following \cite{ju2021sequential}, we now explain how to overcome this problem within the framework of  BFFG. We  will parametrise the functions  
 $g \colon E \to \RR$  by  $g = \psi\circ I$ where $\psi \colon \{0, 1, \ldots, N\} \to \RR$. As the domain of $\psi$ is finite, we can identify $g$ with the vector $\bs{\psi}\in \RR^{N+1}$, where $\bs{\psi}_i := \psi(i)$. Thus  
$g \param{\bs{\psi}}$

The pullback operation is given by
$ (\kappa g)(x) = \sum_{y\in \{0,1\}^N} \kappa(x,y) g(y)$.
The computational cost grows exponentially in $N$. For this reason, consider the ``simpler'' Markov kernel
\begin{equation}\label{eq:tildekappa_ju}
	\tilde\kappa(x,y) = \prod_{i=1}^N \tilde\alpha(x_i, I(x))^{y_i}\left(1-\tilde\alpha(x_i, I(x))\right)^{1-y_i}, 
\end{equation} 
where for $s\in \{0,1,\ldots, N\}$
\begin{equation}\label{eq:tildekappa_ju2} \tilde\alpha(0, s) = \tilde\lambda N^{-1} s, \qquad 
	\tilde\alpha(1,s)  =  1-\bar\gamma.
\end{equation}
Note that under $\tilde\kappa$ all particles evolve conditionally independently with either probability $\tilde\alpha(0,I(x))$ or $\tilde\alpha(1,I(x))$. Crucially for what follows, these probabilities depend on $x$ only via $I(x)$. 

\begin{definition}
Suppose $U_k \sim \mbox{Ber}(\theta_k)$, $1\le k \le N$. If $U_1,\ldots, U_n$ are independent, then the random variable $U=\sum_{k=1}^N U_k$ is said to have the $\mbox{PoiBin}(\theta_1,\ldots, \theta_N)$ distribution. Its probability mass function is given by \[ \PP(U=\ell) = \sum_{x\in \{0,1\}^N} \ind_{\{\sum_{i=1}^N x_i =\ell \}} \prod_{i=1}^N \theta_i^{x_i} (1-\theta_i)^{1-x_i}, \qquad 0\le \ell \le N. \]
\end{definition}

\begin{theorem}\label{thm:ju}
Let $\kappa$ be defined by \eqref{eq:kappa_ju}--\eqref{eq:kappa_ju2} and $\tilde\kappa$ by  \eqref{eq:tildekappa_ju}--\eqref{eq:tildekappa_ju2}.

\begin{enumerate}
\item {\it Computing $\scr{G}$:}
\begin{itemize}
	\item {\it Pullback for $\tilde\kappa$:} if $g  \param{\bs{\psi}_1}$, then $\tilde\kappa g  \param{\bs{\psi}_2}$ with $\bs{\psi_2}$ determined by 
	\[ \psi_2(x) = \EE \psi_1(Z(x)), \]
where $Z(x)$ is distributed as $Z_0(x)+Z_1(x)$, with  $Z_0(x) \sim Bin(N-I(x), \tilde\alpha(0,I(x)))$ and  
$Z_1(x) \sim Bin(I(x), \tilde\alpha(1,I(x))$. Thus,
\begin{align*}	\psi_2(x) &=\sum_{k=0}^N \psi_1(k)\sum_{j=0}^k {N-I(x) \choose k-j} \tilde\alpha(0,I(x))^{k-j} \\ &  \qquad\times (1-\tilde\alpha(0,I(x))^{N-I(x)-k+j} {x \choose j} \tilde\alpha(1,I(x))^j (1-\tilde\alpha(1,I(x))^{I(x)-j}. \end{align*}
This agrees with Equation (32) in \cite{ju2021sequential}. 
	\item {\it Fusion:} Suppose $g_i \param{\bs{\psi}_i}$, $i=1,\ldots, k$.  Then 
\[ \prod_{i=1}^k g_i = \prod_{i=1}^k \psi_i \circ I  \param{\bigcirc_{i=1}^k \bs{\psi}_i},\] where $\bigcirc$ denotes the Hadamard (entrywise) product. 

	\item {\it Initialisation from leaves:} If $v$ is observed at  a leaf  node (i.e.\ an observation), then $g_{\pa(v),v}(x)=\psi_v(I(x)) \param{\bs{\psi}_v}$ with
\[ \bs{\psi}_v(i) = \binom{i}{v} \rho^v (1-\rho)^{i-v}\ind_{\{v \le i \le N\}} ,\qquad i =0,\ldots, N \]
Here, if $i<v$, we can define ${i\choose v}$ arbitrarily, as the indicator ensures that  $\bs{\psi}_v(i)=0$ in that case (alternatively, define ${i\choose v}=0$ for $i<v$ and drop the indicator).

\end{itemize}
\item {\it Sampling under $\PP^\circ$:}
Given state $x$, transition kernel $\kappa$ and $g \param{\bs{\psi}}$, sampling from the guided process can be done in two steps: 
\begin{enumerate}
\item sample  $j^\star$ from  $\{0,\ldots, N\}$ with probabilities\[ p(j)= \frac{\sum_{z\colon I(z)=j} \kappa(x,z) \bs{\psi}_j}{\sum_k\sum_{z\colon I(z)=k} \kappa(x,z) \bs{\psi}_k},\]
\item sample $y$ conditional on $j^\star$ according to 
\[ p(y\mid j^\star) = \frac{\kappa(x,y) \ind\{I(y)=j^\star\}}{\sum_{z\in E}\kappa(x,z) \ind\{I(z)=j^\star\}},\qquad y\in E.\]
\end{enumerate}
Here, we have dropped dependence of $y$ and $j^\star$ on $x$ in the notation.
This agrees with Equation (38) in \cite{ju2021sequential}. 
\item {\it Computation of the weight:} if $g \param{\bs{\psi}}$, then the weight at $x$ is given by 
$w(x) = (\kappa g)(x)/(\tilde\kappa g)(x)$, where 
\[ (\kappa g)(x) = \EE \bs{\psi}_{U(x)},\]
and $U(x) \sim \mbox{PoiBin}(\alpha_1(x),\ldots, \alpha_N(x))$. This agrees with Equation (37) in \cite{ju2021sequential}. 
\end{enumerate}
\end{theorem}

\begin{remark}
Note that when initialising from the  leaves that if $v$ is close to $N$, most $\bs{\psi}_i$ are zero. This sparsity propagates through the backward filtering, as seen from  the pullback and  fusion steps. 
\end{remark}
\begin{remark}
 Sampling from  $y\mid j^\star$ entails sampling from  independent, non-identically distributed Bernoulli variables, conditional on their sum (this is denoted the $\mbox{CondBer}$-distribution in \cite{ju2021sequential}.

\end{remark}

\begin{proof}
We first derive the expression for the pullback. We have 
\begin{align*}
(\tilde\kappa g)(x) &=
\sum_{y\in \{0,1\}^N} \tilde\kappa(x,y) \psi_1(I(y)) \\ &=\prod_{i=1}^N \tilde\alpha(x_i, I(x))^{y_i}\left(1-\tilde\alpha(x_i, I(x))\right)^{1-y_i} \psi_1(I(y))\\ & = \prod_{i=1}^N \tilde\alpha(0,I(x))^{y_i} (1-\tilde\alpha(0,I(x))^{1-y_i} \psi_1(I(y))
, 	\end{align*}	
Hence
\[ (\tilde\kappa g)(x)= \EE\, \psi_1(Z(x)) = \sum_{k=0}^N \psi_1(k) \PP(Z(x)=k)\]
where $Z(x)$ is the number of successes in $N$ independent Bernoulli trials, where the $i$-th Bernoulli trial has success probability $A_i:=\tilde\alpha(x_i, I(x))$. That is, $Z\sim PoisBin\left(\{A_i\}_{i=1}^N\right)$. 
Now note that $A_i$ can only takes values $\tilde\alpha(0,I(x))$ and $\tilde\alpha(1,I(x))$. Hence, we consider $I(x)$ Bernoulli random variables with success probability $\tilde\alpha(1,I(x))$, and $N-I(x)$ Bernoulli random variables with success probability $\tilde\alpha(0,I(x))$. Then indeed $Z(x)$ is distributed as the sum of $Z_0(x)$ and $Z_1(x)$. The final expression simply follows from the probability mass function for the convolution of two discrete random variables. 

The fusion operation follows directly. 

The expression for the initialisation from the leaves follows directly from $V \sim \mbox{Bin}(I(x),\rho)$. 

To sample the guided process,
 \[ \kappa^\circ(x,y) \propto \kappa(x,y) g(y) = \frac{\kappa(x,y) \psi(I(y))}{\sum_{z\in E} \kappa(x,z) \psi(I(z))} \]
Then note that (the summation is over $j\in \{0,\ldots, N\}$ and $k\in \{0,\ldots, N\}$ ), 
\begin{align*} & p(y)= \sum_j p(y\mid j) p(j)=\sum_j  \frac{\kappa(x,y) \ind\{I(y)=j\}}{\sum_{z\in E}\kappa(x,z) \ind\{I(z)=j\}} \frac{\sum_{z\colon I(z)=j} \kappa(x,z) \psi(j)}{\sum_k\sum_{z\colon I(z)=k} \kappa(x,z) \psi(k)}\\ &=\frac1{\sum_k\sum_{z\colon I(z)=k} \kappa(x,z) \psi(k)} \sum_j\frac{\kappa(x,y) \sum_{z\colon I(z)=j) }\kappa(x,z) \psi(j) \ind\{I(y)=j\}}{\sum_{z\in E}\kappa(x,z) \ind\{I(z)=j\}}\\ & = \frac{\kappa(x,y) \psi(I(y))}{\sum_{z} \kappa(x,z) \psi(I(z))} \sum_j\frac{\sum_{z\colon I(z)=j) }\kappa(x,z)  \ind\{I(y)=j\}}{\sum_{z\in E}\kappa(x,z) \ind\{I(z)=j\}}.
	\end{align*}
Note that the last term on the third line equals $1$, since
\[ \sum_j \ind\{I(y)=j\}  \frac{\sum_{z\colon I(z)=j) }\kappa(x,z)  }{\sum_{z\colon I(z)=j}\kappa(x,z) }=1.\]
Therefore, $p(y) =\kappa^\circ(x,y)$. 
 
 Finally, we have
 \begin{align*}
 	(\kappa g)(x) & = \sum_{y\in \{0,1\}^N} \kappa(x,y) g(y) \\ & =
 	 \sum_{y\in \{0,1\}^N} \prod_{i=1}^N \alpha_i(x)^{y_i} (1-\alpha_i(x))^{1-y_i} \psi(I(y)) \\ & =
 	 \sum_{k=0}^N  \sum_{y\in \{0,1\}^N} \ind_{\{\sum_{i=1}^N y_i=k\}}\prod_{i=1}^N \alpha_i(x)^{y_i} (1-\alpha_i(x))^{1-y_i} \psi\left(\sum_{i=1}^N y_i\right)\\ & =
 	 \sum_{k=0}^N \bs{\psi}_k \sum_{y\in \{0,1\}^N} \ind_{\{\sum_{i=1}^N y_i=k\}}\prod_{i=1}^N \alpha_i(x)^{y_i} (1-\alpha_i(x))^{1-y_i} \\ & = \sum_{k=0}^N \PP(U(x)=k) \bs{\psi}_k.
 \end{align*}

\end{proof}

We refer to \cite{ju2021sequential} for methods to efficiently perform the steps in Theorem \ref{thm:ju} and application within the context of controlled Sequential Monte Carlo (\cite{Heng2020}).

\section{Examples of guided processes for continuous edges}
\label{sec:examples_continuous}
\subsection{Stochastic differential equations}\label{ex:sde}

Assume a directed tree where on each edge $e=(s,t)$ not pointing to a leaf vertex, the Markov process $X$ is defined as the solution to the stochastic differential equation (SDE) \begin{equation}\label{eq:sde}
	\dd X_u = b(u,X_u) \dd u + \sigma(u,X_u) \dd W_u,\qquad u\in [0,\tau_e]
\end{equation} 
As transition densities are not known in closed form, computing the BIF is infeasible. However, for the SDE
\begin{equation}\label{eq:auxiliary_sde} \dd \tilde X_u = (B(u) \tilde X_u + \beta(u)) \dd u + \tilde\sigma(u) \dd W_u \end{equation} 
this is possible. Hence, suppose that on each edge $e$ we choose the triplet of maps  $(B_e, \beta_e, \tilde\sigma_e)$. If follows from the results in \cite{mider2021continuous} that for $u\in (0,\tau_e]$ we have 
$g_u(x) = \exp\left(c(u) + F(u)^\T x -\frac12 x^\prime H(u) x\right)$, for scalar-valued $c$, vector-valued $F$ and matrix-valued $H$. As such, we write $g_u \param{c(u),F(u),H(u)}$. 

\begin{theorem}\label{thm:sde}
Assume a directed tree with leaf observatrions $x_v \mid x_{\pa(v)} \sim N(L_v x_{\pa(v)}, \Sigma_v)$. On  each edge $e$, the process evolves according to \eqref{eq:sde}. 
\begin{enumerate} 
\item {\it Computing $\scr{G}$:}
\begin{itemize}
	\item {\it Pullback for $\tilde\kappa$:} if on an edge $e=(s,t)$ we are given $g_t \param{c(t),F(t),H(t)}$, then  for $u\in (0,\tau_e]$, $g_u \param{c(u),F(u),H(u)}$, where 
 $(c(u), F(u), H(u))$ are solved backwards from 
\begin{equation}
  \label{eq:backwardODE}
    \begin{split}
        \dd H(u) &= \left( -B(u)^\T H(u) - H(u)B(u)+H(u)\tilde{a}(u)H(u) \right)\dd u,\\
        \dd F(u) &= \left( -B(u)^\T F(u) + H(u)\tilde{a}(u)F(u) + H(u)\beta(u) \right)\dd u,\\
        \dd c(u) &= \left( \beta(u)^\T F(u) + \frac{1}{2}F(u)^\T \tilde{a}(u)F(u) - \frac{1}{2}\trace\left(H(u)\tilde{a}(u)\right) \right)\dd u,
    \end{split},
  \end{equation}
 subject to $(c(t), F(t), H(t))$. Here, $\tilde a=\tilde\sigma \tilde\sigma^\T$. 

	\item {\it Fusion:} if $g_i \param{c_i, F_i, H_i}$, then $\prod_i g_i \param{\sum_i c_i, \sum_i F_i, \sum_i H_i}$. 
	\item {\it Initialisation from leaves:} 
\begin{equation}\label{eq:sde_initleaves} g_{\pa(v),v} \param{ \log \phi(v; 0, \Sigma_v), v^\T \Sigma_v^{-1} L_v ,  L_v^\T \Sigma_v^{-1} L_v}.   \end{equation}
\end{itemize}
\item {\it Sampling under $\PP^\circ$:} on the edge $e=(s,t)$ the guided process evolves according to the SDE
\begin{equation}\label{eq:Xcirc_sde} \dd X^\circ_u = \left(b(u,X^\circ_u) + a(u,X^\circ_u) (F(u)-H(u) X^\circ_u) \right) \dd u + \sigma(u,X^\circ_u) \dd W_u,\qquad u\in [0,\tau_e],  \end{equation}
subject to $X^\circ_0=x$, with $x$ the state of the process at vertex $s$. 
\item {\it Computation of the weight:} $\log w_{s,t}(X^\circ)=\int_0^{\tau_e} \frac{({\cL}-\tilde{{\cL}})  g}{g}(u, X^\circ_u) \dd u$. 
\end{enumerate}
\end{theorem}
\begin{proof}
This is a consequence of Theorem 2.5 in \cite{mider2021continuous}.
\end{proof}

\begin{remark}
The weight is obtained by integrating 
\[  \frac{({\cL}-\tilde{{\cL}})  g}{g}  = \sum_i (b_i-\tilde{b}_i) \frac{\partial_i  g}{g} + \frac12 \sum_{i,j} (a_{ij}-\tilde{a}_{ij}) \frac{\partial^2_{ij}  g}{g} \]
 over $[0,\tau_e]$. Here, 
 we have omitted arguments $(u,X^\circ_u)$ from the functions and $\partial_i$ denotes the partial derivative with respect to $x_i$. If we set $r_i = (\partial_i  g)/g$, then $(\partial_{ij}  g)/  g = \partial_j  r_i +  r_i  r_j$ and therefore
\[ \frac{({\cL}-\tilde{{\cL}})  g}{g} = \sum_i (b_i-\tilde{b}_i) r_i + \frac12 \sum_{i,j} (a_{ij}-\tilde{a}_{ij}) \left( \partial_j  r_i +  r_i  r_j\right). \]
This expression coincides with the expression for the likelihood given in Proposition 1 of \cite{schauer2017guided}, but the proof given here is much shorter. 
\end{remark}
\begin{remark}
Note that the solving the ODEs in the pullback operation scales cubically in the dimension of the diffusion process. Improved scaling can be obtained by introducing sparsity in $B$ and/or $\tilde\sigma$. 
\end{remark}


\subsection{Continuous time Markov chains  with  countable state space}\label{subsec:ctmc}

Let $X$ denote a continuous time Markov process taking values in  the countable  set $E$.  Let $\cQ = (q(x,y),\, x, y \in E)$ be a $Q$-matrix, i.e.\ its elements $q(x,y)$ satisfy $q(x,y) \ge 0$ for all $x\neq y$ and $\sum_y q(x,y) =0$. Define $c(x)=-q(x,x)$. If $\sup_x c(x)<\infty$, then $\cQ$ uniquely defines a continuous time Markov chain with values in $E$. Its transition probabilities $p_u(x,y)$ are continuously differentiable in $u$ for all $x$ and $y$ and satisfy Kolmogorov's backward equations:
\[ \frac{\partial}{\partial u} p_u(x,y) = \sum_x q(x,z) p_u(z,y) \]
(cf.\ \cite{Liggett2010}, Chapter 2, in particular Corollary 2.34). The infinitesimal generator  acts upon functions $f\colon E \to \RR$ by
\begin{equation}\label{eq:genCTTMchain} \scr{L} f(x) = \sum_{y \in E} q(x,y) \left(f(y) - f(x)\right) = \sum_{y\in E} q(x,y) f(y),\qquad u\in [0,\infty), \quad x \in E. \end{equation}
In the specific setting where $E$ is finite with states labeled by $1, 2, \ldots, R$, $\cQ$ is an $R\times R$-matrix. Then, by evaluating $\scr{L}f(x)$ for each $x\in E$ we can identify  $\scr{L} f$ by $\cQ \bs{f}$,
where $\bs{f}$ is the column vector with $i$-th element $\bs{f}_i =f(i)$. 
\medskip

Next, we explain how the law of the process $X$ changes under the $h$-transform. Thus, 
suppose  $h\colon [0,\infty) \times E \to \RR$ is specified. This function can be used to define a change of measure as detailed in Section \ref{subsec:continuous}. It induces a guided process $X^\circ$ for which the generator of the space-time process can be derived by combining Equation \eqref{eq:infgen-guided} and  the preceding display. Some simple calculations reveal that 
 the generator of $X^\circ$ acts upon functions $f\colon [0,\infty) \times  E \to \RR$ as
\begin{equation}\label{matrixcirc}
(\cL^\circ_u f)(x)= \sum_{y \in E} q(x,y) (f(y)-f(x))  \frac{h(u, y)}{h(u,x)}.
\end{equation}
Therefore, comparing with \eqref{eq:genCTTMchain}, we see that the guided process $X^\circ$ is a  time-inhomogenuous Markov process $X^\circ$ with time dependent generator matrix $\cQ_u^\circ = (q_u^\circ(x,y),\, x, y \in E)$, where
\[ q_u^\circ(x,y) =  q(x,y)  \frac{h(u, y)}{h(u,x)} \quad \text{if} \quad y\neq x\] 
and $q_u^\circ(x,x) = 1- \sum_{y\neq x} q^\circ_u(x,y)$. 

For finite-state Markov chains transition densities can be computed via $(p_u(x,y),\, x, y \in E) = e^{u \cQ}$ and this would immediately give $h$, implying no need for constructing a guided process via an approximation to Doob's $h$-transform. While strictly speaking this is true, in case $|E|$ is large, numerically approximating  the matrix exponential can be computationally demanding. In that case, by simplifying the dynamics of the Markov process, one may choose an approximating continuous time Markov process $\tilde X$ (on $E$) where $\tilde \cQ$ is block diagonal, as this simplifies evaluation of matrix exponentials.
 
Hence, let  $\tilde X$ be a  continuous time Markov process on $E$ with $Q$-matrix $\tilde\cQ = (\tilde q(x,y),\, x, y \in E)$. Then the guided process has $Q$-matrix $\mathcal{Q}^\circ=(q^\circ(x,y),\, x, y \in E)$, where  for $y\neq x$, $q_u^\circ(x,y) =  q(x,y) g(u, y)/g(u,x)$. 
If the guided process is simulated on an edge $e$, its log-weight equals 
\[ \int_0^{\tau_e} \frac{({\cL}-\tilde{{\cL}}) g}{g} (u, X^\circ_u) \dd u =  \int_0^{\tau_e} \frac{ \sum_{y \in E} (q(X^\circ_{u},y)-\tilde{q}(X^\circ_{u},y)) g(u,y)}{g(u,X^\circ_u)} \dd u.\]

 We summarise the key ingredients of BFFG in case the state space is given by  $E=\{1,\ldots, R\}$. Suppose $g \colon E \to \RR$.  Then we can idenfity $g \param{\bs{g}}$, with $\bs{g} \in \RR^R$ given by $\bs{g}_i=g(i)$.

  \begin{theorem} Assume a directed tree, where each  edge not ending in a leave node is a continuous edge. On any such edge $e=(s,t)$, the process evolves as a continuous-time Markov chain  with generator matrix $\cQ$ over time interval $[0,\tau_e]$. 
\begin{enumerate}
\item {\it Computing $\scr{G}$:}
\begin{itemize}
	\item {\it Pullback for $\tilde\kappa$:} if $g_t \param{\bs{g}_t}$, then for $u\in (0,\tau_e]$, $g_u \param{e^{\tilde \cQ(\tau_e-u)}\bs{g}_t}$. 
		\item {\it Fusion:} if $g_i  \param{\bs{g}_i}$, $i = 1, \dots, k$, then
\[  \prod_{i=1}^k g_i   \param{\bigcirc_{i=1}^k \bs{g}_i},\] where $\bigcirc$ denotes the Hadamard (entrywise) product.

	\item {\it Initialisation from leaves:} Upon observing state $i$ at a leaf $v$ set  $g_{\pa(v),v} \param{\bs{e}_i}$, where $\bs{e}_i$ is the $i$th unit vector in $\RR^R$.

\end{itemize}
\item {\it Sampling under $\PP^\circ$:}  on the edge $e$ the guided process has  time dependent generator matrix $\cQ_u^\circ = (q_u^\circ(x,y),\, x, y \in E)$, where
\[ q_u^\circ(x,y) =  q(x,y)  \frac{g(u, y)}{g(u,x)} \quad \text{if} \quad y\neq x\] 
and $q_u^\circ(x,x) = 1- \sum_{y\neq x} q^\circ_u(x,y)$. 

\item {\it Computation of the weight:}
\[ \log w_{s,t}(X^\circ) = \int_0^{\tau_e} \frac{ \sum_{y \in E} (q(X^\circ_{u},y)-\tilde{q}(X^\circ_{u},y)) g(u,y)}{g(u,X^\circ_u)} \dd u.\]
\end{enumerate}
\end{theorem}

\subsection{Model by \cite{stoltz2021bayesian}}\label{subsec:stoltz}

Consider $N$ diploid individuals (assume red/green, or, more typically, alleles $a$ and $A$). At time $i$, let $X_i$ denote the number of red alleles. The classical Wright-Fisher model with mutation is given by
\[ X_{i+1} \mid X_i \sim Bin(2N, (1-u)X_i/(2N) + v(1-X_i/(2N)), \]
where $u$ is probability of mutating from red to green and $v$ is the probability of mutating from green to red. 
A diffusion approximation (using a rescaling of time) gives
\begin{equation}\label{eq:difflimit} \dd X_u = \left(\beta_1(1-X_u) + \beta_2 X_u\right) \dd t + \sqrt{X_u(1-X_u)} \dd W_u, \end{equation}
where $\beta_1= 2Nu$, $\beta_2=2Nv$.
Denote $b(x) = \beta_1(1-x) + \beta_2 x$ and $a(x)=x(1-x)$. 
\cite{stoltz2021bayesian} consider a stochastic process evolving on a tree, where along branches the process evolves according to \eqref{eq:difflimit} and at the leaves one observes $v_i \sim Bin(n_i, x_i)$.

\medskip

\cite{stoltz2021bayesian} explain an efficient procedure to compute the backward information filter. The basic idea is to expand $g$ at any time instance as a linear combination of shifted Chebyshev polynomials. Hence, 
\begin{equation}\label{eq:gseries}
	g(x) = \sum_{k=0}^K \lambda_k \psi_k(x), \qquad x\in [0,1],
\end{equation} 
where $\psi_k(x) = T_k(2x-1)$ with $\{T_k\}$ the Chebyshev polynomials of the first kind. 
Define 
$\scr{S}_k = \mbox{span}(\psi_0,\ldots, \psi_K)$. 
If $g$ is expanded as in \eqref{eq:gseries}, then we write $g\param{\lambda_0,\ldots, \lambda_K}$
Alternatively, we can parametrise by the values of $g$ in Che\-by\-shev-\-Lo\-bat\-to points on $[0,1]$, which we denote by $x_0, x_1, \ldots, x_K$.  In this case we write $ g \parambar{g(x_0), g(x_1),\ldots, g(x_K)}$. If $g \param{\lambda_0,\ldots, \lambda_K}$ is given, it is trivial to obtain $\parambar{g(x_0), g(x_1),\ldots, g(x_K)}$. The reverse operation can be done  in $O(K\log K)$-time using the FFT (details are given in in the appendix of \cite{stoltz2021bayesian}).
With these facts, we can fully describe the key steps for backwards filtering.

\begin{theorem} Suppose $X$ is a a process on a directed tree evolving according to the SDE \eqref{eq:difflimit} over each edge $e$ for time $[0,\tau_e]$.
\begin{enumerate}
\item {\it Computing $\scr{G}$:}
\begin{itemize}
	\item {\it Pullback for $\tilde\kappa$:} suppose over the edge $(s,t)$ the process $X$ defined by \eqref{eq:difflimit} evolves for time $u\in [0,\tau_e]$. If $g_t\param{\lambda_0,\ldots, \lambda_K}$, then 
\[ g_u \param{\lambda_0(u),\ldots, \lambda_K(u)},\qquad u\in (0,\tau_e],\]
where  $\bs{\lambda}(u)=[\lambda_0(u), \ldots, \lambda_K(u)]$, $u\in (0,\tau_e]$ satisfies the ordinary differential equation
\[ \partial_t \bs{\lambda}(u)+ Q \bs{\lambda}(u) = 0, \qquad \bs{\lambda}(\tau_e)=[\lambda_0,\ldots, \lambda_K].\]
The matrix $Q$ is upper-triangular and determined by \eqref{eq:defQ}. 
	\item {\it Fusion:} if $g_i \parambar{g_i(x_0),\ldots, g_i(x_K)}$, then 
\[ \prod_i g_i \parambar{\prod_i g_i(x_0),\ldots,\prod_i g_i(x_K)}. \]
	\item {\it Initialisation from leaves:} Upon observing at a leaf $v \sim Bin(n, x)$ set $g\parambar{o(x_0), o(x_1),\ldots, o(x_K)}$ with $o(x)={n \choose v} x^v (1-x)^{n-v}$. 

\end{itemize}
\item {\it Sampling under $\PP^\circ$:} not considered in \cite{stoltz2021bayesian}.
\item {\it Computation of the weight:} not considered in \cite{stoltz2021bayesian}.

\end{enumerate}
\end{theorem}
\begin{proof}
	
The pullback operation consists of  solving the Kolmogorov backward equation
\begin{equation}\label{eq:kolmbackw}
	\scr{L} g_u(x) +\partial_u g_u(x)=0,\qquad g_{\tau_e} = g_T, . 
\end{equation} 
where $u\in (0,\tau_e]$. 
The ansatz is that 
$g_u(x)= \sum_{k=0}^K \lambda_k(u) \psi_k(x)$. Plugging this expression into \eqref{eq:kolmbackw} gives
\[ \sum_{k=0}^K \left[ \partial_u \lambda_k(u) \psi_k(x) + b(x)\lambda_k(u) \frac{\partial \psi_k(x)}{\partial x} + \frac12 a(x) \lambda_k(u) \frac{\partial^2 \psi_k(x)}{\partial x^2} \right]=0. \]
By choice of the drift and diffusivity in \eqref{eq:difflimit},  for each $k$, the map $x\mapsto b(x)\frac{\partial \psi_k(x)}{\partial x} + \frac12 a(x)  \frac{\partial^2 \psi_k(x)}{\partial x^2}$ is in $\scr{S}_k$. Therefore,  there exist coefficients $Q_{ik}$, $0\le i \le K$ (with $Q_{ik}=0$ for $i>k$) such that
\begin{equation}\label{eq:defQ}
	b(x)\frac{\partial \psi_k(x)}{\partial x} + \frac12 a(x)  \frac{\partial^2 \psi_k(x)}{\partial x^2} = \sum_{i=0}^K Q_{ik} \psi_i(x).
\end{equation}   
This implies 
\[ \sum_{k=0}^K \partial_u\lambda_k(u) \psi_k(x) + \sum_{i=0}^K \sum_{k=0}^K  \lambda_k(u) Q_{ik} \psi_i(x)=0\]
which can be rewritten to 
\[ \sum_{k=0}^K \left( \partial_u\lambda_k(u) + (Q \bs{\lambda}(u))_k \right)\psi_k(x)=0. \]
The result follows by linear independence of the basis functions  $\{\psi_k\}_{k=0}^N$. 

\end{proof}

\section{Extension to a Directed Acyclic Graph}\label{sec:extension_to_DAG}

In this section we extend our approach from a directed tree to a true DAG. Whereas on the former each vertex has a unique parent vertex, on a DAG a vertex can have multiple parent vertices.
Note that on a DAG, contrary to a directed tree, conditioning on the values at leaf-vertices changes the dependency structure (there is a conditional process $X^\star$ defined on the vertex set, but its dependency graph is different from that of $X$, in particular there is
no meaningful transition $\kappa^\star_{\pa(s), s}$ for all $s$). Also, there is no BIF-type recursion that enables computing the likelihood $h_r(x_r)$.   We retain the definition of the guided process as in Definition \ref{defn: guided process}. In particular, unlike the conditioned process, the guided process  has the same dependency structure as the (unconditional) forward process. Crucially this means a sampler for the guided process has a similar complexity as a sampler for  $X$, predisposing guided process for sampling based inference. While this enhances automatic translation of sampling unconditional to conditional processes, it need not necessarily lead to computational gains in specific examples when compared to approaches that break the dependency structure. 

On a DAG, with not much effort we can get a similar statement as Theorem \ref{thm:lr_xstar_xcirc_tree} by redefining the denominator of the  weights for $s\in \cS$. 

\begin{theorem}\label{thm:lr_xstar_xcirc_dag}
Theorem \ref{thm:lr_xstar_xcirc_tree} applies to a DAG as well if 
the weights are redefined by
\begin{equation}\label{eq:dagweights}
w_{\pa(s), s}(x) =
\begin{cases}
\int g_s(y)\kappa_{\pa(s), s}(x,\dd y)   \Big/ \prod_{u \in \pa(s)} g_{u, s}(x_u)    & \text{if} \quad s\in \cS\\
h_{\pa(v) , v}(x) \Big/ g_{\pa(v), v}(x) &\text{if} \quad s\in \cV. 
\end{cases}
\end{equation} 
\end{theorem}
Note the change in weights is only in the denominator for $s\in \cS$. 
\begin{proof}
The expression for $(\dd \PP^\star/\dd \PP)(X_\cS)$ now follows from Bayes' theorem. 
The remainder of the proof is the same as for Theorem \ref{thm:lr_xstar_xcirc_tree}, except that \eqref{eq:cancellation} takes the form
\[ \frac{ \prod_{s\in\cS}\prod_{t\in \ch(s)} g_{s, t}(X_s) }{ \prod_{s\in\cS}\prod_{u \in \pa(s)} g_{u, s}(X_u)} = \frac{\prod_{v\in \cV} g_{\pa(v), v}(X_{\pa(v)}) }{g_0(x_0)}. \]
\end{proof}

\begin{figure}
\begin{center}
\begin{tikzpicture}
	\tikzstyle{empty}=[fill=white, draw=black, shape=circle,inner sep=1pt, line width=0.7pt]
	\tikzstyle{solid}=[fill=black, draw=black, shape=circle,inner sep=1pt,line width=0.7pt]
	\begin{pgfonlayer}{nodelayer}
		\node [style=solid,label={$1a$},] (-1) at (-4.75, 1) {};
		\node [style=solid,label={$1b$},] (0) at (-4.75, -1) {};
		\node [style=solid,label={$2$}] (1) at (-3, 0) {};
		\node [style=solid,label={$4a$}] (2) at (-0.25, 1) {};
		\node [style=solid,label={$4b$}] (3) at (-0.25, -1) {};
		\node [style=empty,label={$5a$}] (4) at (1.75, 1) {};
		\node [style=empty,label={$5b$}] (5) at (1.75, -1) {};
		
	\end{pgfonlayer}
	\begin{pgfonlayer}{edgelayer}
		\draw [style=edge] (-1) to (1);
		\draw [style=edge] (0) to (1);
		\draw [style=edge] (1) to (2);
		\draw [style=edge] (1) to (3);
		\draw [style=edge] (2) to (4);
		\draw [style=edge] (3) to (5);
	\end{pgfonlayer}
\end{tikzpicture}
\end{center}
\caption{A  DAG with two leaves and two roots.}
\label{fig:collider}
\end{figure}
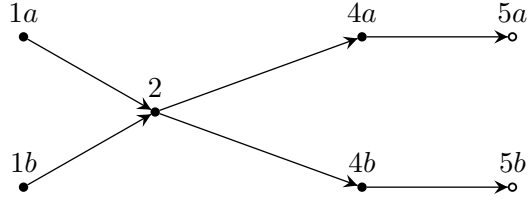
 
Consider the DAG in  Figure \ref{fig:collider}. The rightmost picture shows the situation corresponding to a collider in graphical models. Here  $\kappa_1 = \kappa_{1a}\otimes \kappa_{1b}$ has product form, but  $\kappa_2$ does not. The transformed kernel $\kappa_1^\star$ typically won't have product form, the states become entwined \citep{jacobs2019structured}.   

We explain a way to deal with  a kernel with multiple parents, say $\kappa_{\pa(s), s}$ in backward filtering.    
Its input $x = \{x_u,\, u\in\pa(s)\}$ is a concatenation of values of the process at the parent vertices. For an approximation $\tilde\kappa_{\pa(s), s}$ of $\kappa_{\pa(s), s}$, we have pullback 
\begin{equation}\label{eq:pullback multipleparents1}
g(x_{\pa(s)}):=(\tilde\kappa_{\pa(s), s} g)(x_{\pa(s)}).\end{equation} To backpropagate this function, we want to have a function in product form $\prod_{u\in \pa(s)} g_u(x_u)$ to each parent individually.  To this end, write $x_{\pa(s)}=(x_1,\ldots, x_d)$ and suppose the density of $x_{\pa(s)}$ is given by $\pi(x_1,\ldots, x_d)$. Then the joint distribution of $X_{\pa(s)}$ and its leaf descendants is  proportional to $\pi(x_1,\ldots, x_d)g(x_1,\ldots, x_d)$. Let $\pi_i$ denote the marginal distribution of $X_i$. Denote $x_{-i}$ the vector $x$ without its $i$-th component. Recall that for $\lambda$-probability densities $p$ and $q$, the Kullback-Leibler divergence of $p$ to $q$ is given by $\on{KL}(p,q) = \int p \log p/q \dd \lambda$. 
\begin{proposition}\label{prop: kl_backward}
The minimiser of $\on{KL}(\pi g, \prod_{i=1}^d \pi_i g_i)$ over $g_1,\ldots, g_d$, subject to \[\int g_i(x_i) \pi_i(x_i) \dd \lambda(x_i) =c_i, \qquad  1\le i \le d \] is given by 
\[ g_i(x_i) = c_i^{-1} \EE_\pi [g(X_1,
\ldots, X_d \mid X_i=x_i]. \]
\end{proposition}
The proof is given in Section \ref{sec:remaining proofs}. 

As a consequence to this proposition (with $c_j=1$ for all $j$) we define $g_u$ by 
\begin{equation}\label{eq:pullback multipleparents2}
		g_u(x_u)   = \int g(x_{\pa(s)}) \pi(x_{-u}\mid x_u) \dd x_{-u},\qquad u\in \pa(s).	
\end{equation}
Upon defining 
\begin{equation}\label{eq:kappa_u_to_s}	 
	\tilde\kappa_{u, s}(x_u, \dd y) := \int \tilde\kappa_{\pa(s), s}(x_{\pa(s)}, \dd y) \pi(x_{-u}\mid x_u) \dd x_{-u}, 
\end{equation}
it follows that we can equivalently define $g_u$ in terms of these kernels by
\begin{align*}
g_u(x_u) =& \iint  \tilde\kappa_{\pa(s), s}(x_{\pa(s)}, \dd y) g(y) \dd y\, \pi(x_{-u}\mid x_u) \dd x_{-u} \\
	= & \int \tilde\kappa_{u, s}(x_u, \dd y) g(y) \dd y = (\tilde\kappa_{u, s}g)(x_u)
\end{align*}
(use Fubini's theorem). 
Here, the  prediction kernels $\tilde\kappa_{u, s}(x_u,\cdot)$, $u \in \pa(s)$ serve as \emph{a priori predictor} of $X_s$ if only $x_u$ is known. (From this perspective,  $\kappa_{\pa(s), s}(x,\cdot)$ is the prior predictive distribution of $X_s$ if  $x = \{x_u,\, u\in\pa(s)\}$ is known.)

\begin{definition}
	\label{defn:pullback_dag}
Assume $\vert\pa(s)\vert\ge 2$. 	
To the kernel $\kappa_{\pa(s), s}$ we define the pullback-kernel  $\bar\kappa_{\pa(s), s}$  	 induced by $\pi(x_{\pa(s)})$ 
  by its action on functions $g \in \bB(S_{s})$ 
  \[(\bar\kappa_{\pa(s), s}g)(x) = \prod_{u\in \pa(s)} (\tilde\kappa_{u, s} g)(x_u). \]
 Here, the Markov kernels $\tilde\kappa_{u, s}$, $u\in \pa(s)$ are defined in Equation \eqref{eq:kappa_u_to_s}.
\end{definition}

 The operator $\bar\kappa_{\pa(s), s}$ is {\it not} a  Markov kernel operator and therefore we denote it with a bar rather than a tilde.

\section{Numerical Example: diffusion process on a tree}\label{sec:sde_tree}

We consider a directed tree where a along each edge a diffusion process generated by an stochastic differential equation (SDE) evolves. 
As the transition densities of the process are intractable, we adopt the approach we have outlined:
\begin{enumerate}
\item on each of the edges we define a function $g$ by backward filtering;
\item the process $X^\circ$ is forward simulated.
\end{enumerate}
We illustrate the described methods using two examples of SDEs on directed trees. The backward filtering and forward guiding is based on Theorem \ref{thm:sde}. 
All experiment and figures were produced using the Hyperiax framework \url{https://github.com/computationalevolutionarymorphometry/hyperiax/}.

\subsection{An MCMC algorithm for parameter estimation}\label{sec:num}

We start with a low dimensional example where the process evolves on each branch  according to the SDE
\begin{equation}\label{eq:tanh_sde}
	\dd X_t = \tanh.\left( \begin{bmatrix} -\th_0 & \th_0 \\ 
 \th_1 & -\th_1 \end{bmatrix} X_t \right) \dd t + \begin{bmatrix} \sigma_0 & 0 \\ 0 & \sigma_1 \end{bmatrix} \dd W_t,
\end{equation}   where ``$.$'' appearing in the drift means that the $\tanh$ function is applied coordinatewise.
 Here $\bs{\th}:=(\theta_0, \theta_1, \sigma_0, \sigma_1) \in (0,\infty)^4$. Note that $\theta_0$ is the force at which the first component is drawn towards the second component (and $\theta_1$ vice versa). We forward simulated from the model with parameters $(\theta_0, \theta_1, \sigma_0, \sigma_1)=(0.0,0.65,0.1,0.4)$ on a tree with 5 levels and 121 nodes of which 81 are leaf nodes. The edge lengths are uniform random sampled in the interval $[1.2,2.2]$. The resulting sampled values for both coordinates are visualised in Figure \ref{fig:sdetreepaths}.
 \begin{figure}
\begin{center}
\includegraphics[scale=0.35,trim=20 0 0 20,clip=true]{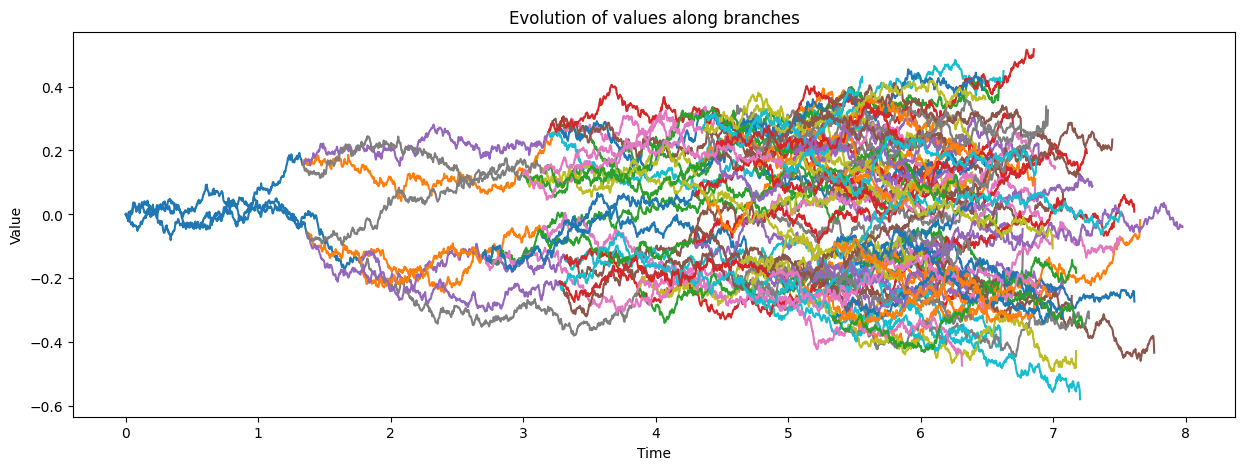}	
\includegraphics[scale=0.35,trim=20 0 0 20,clip=true]{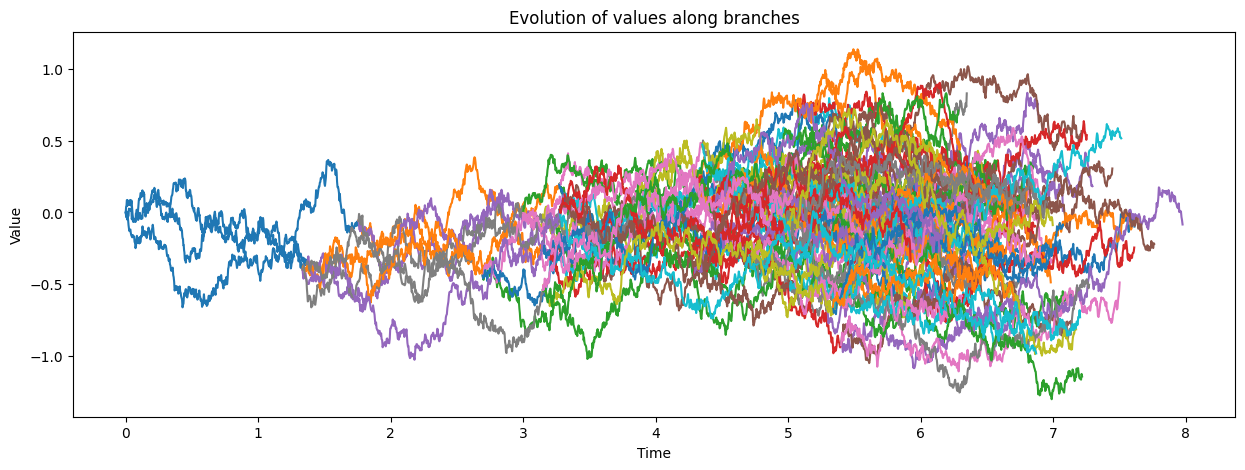}	
\caption{Forward simulated paths from \eqref{eq:tanh_sde} with $(\theta_0, \theta_1, \sigma_0, \sigma_1)=(0.0,0.65,0.1,0.4)$ on a 5 level tree with 121 nodes of which 81 are leaf nodes. Only the values at the leaf nodes are observed.  }  \label{fig:sdetreepaths}
\end{center}
\end{figure}

We assume, as throughout, that only the values at the leaf-vertices are observed. The leaf observations happens with added uncorrelated Gaussian noise in each component with variance $1e-3$. This is incorporated into the backwards filter by setting $\Sigma_v=1e-3$ in Equation \eqref{eq:sde_initleaves}. The tree-structure itself is assumed known though. We aim to estimate the parameters $\bs{\th}$.  We employ flat priors and use an MCMC-algorithm that iteratively updates the unobserved paths conditional on $\bs{\theta}$ and the observations, and $\bs{\theta}$ conditional on the unobserved paths. Elements of $\bs{\th}$ were updated using random-walk Metropolis-Hastings steps. The missing paths were updated using the BFFG-algorithm, where $\tilde{X}$ is chosen as in \eqref{eq:auxiliary_sde}, with 
\[ B = \begin{bmatrix} -\th_1 & \th_1 \\ 
 \th_2 & -\th_2 \end{bmatrix} \qquad \beta =\begin{bmatrix} 0 \\ 0 \end{bmatrix} \qquad \tilde\sigma = \begin{bmatrix} \sigma_1 & 0 \\ 0 & \sigma_2 \end{bmatrix}. \]
 Rather than the missing paths, we updated the innovations driving the SDE using a preconditioned Crank-Nicolson (pCN) update--defined in \eqref{eq:pCN}-- with $\lambda=0.9$, for otherwise the resulting chain would be reducible (the pCN updates are just as in \cite{mider2021continuous}, see \cite{roberts2001inference} for the necessity of updating innovations rather than paths directly). A detailed description of the algorithm used is given in the appendix in Section \ref{app:alg}.
Traceplots after running the algorithm for $20\_000$ iterations are shown in Figure \ref{fig:traceplots} and corresponding density plots after a burn in of $2000$ iterations are shown in Figure \ref{fig:densities}. The acceptance rate rate was $0.58$. It is remarkable that especially $\theta_0$ and $\sigma_0$ can be recovered quite well. Note that this example serves to illustrates BFFG; the chosen MCMC-sampler for updating the parameter $\theta$ is about the simplest choice possible and more efficient proposals and samplers can be exploited. 

 \begin{figure}
\begin{center}
	\includegraphics[scale=0.5,trim=0 0 360 0,clip=true]{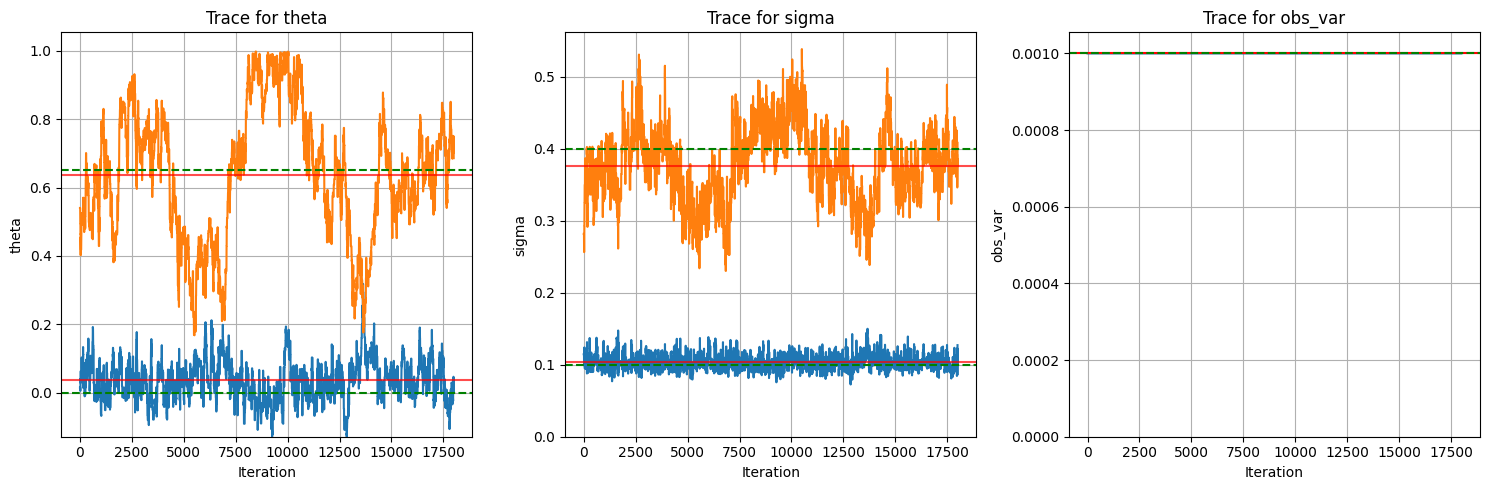}
	\caption{Traceplots for the parameters $\theta_0,\theta_1,\sigma_0,\sigma_1$. Traceplots for $\theta_0$ and $\sigma_0$ are in blue; traceplots for $\theta_1$ and $\sigma_1$ are in orange. Green horizontal lines indicate true valeus; red horizontal lines show posterior mean after removing burnin samples.\label{fig:traceplots}}
\end{center}
\end{figure}
 \begin{figure}
\begin{center}
	\includegraphics[scale=0.37]{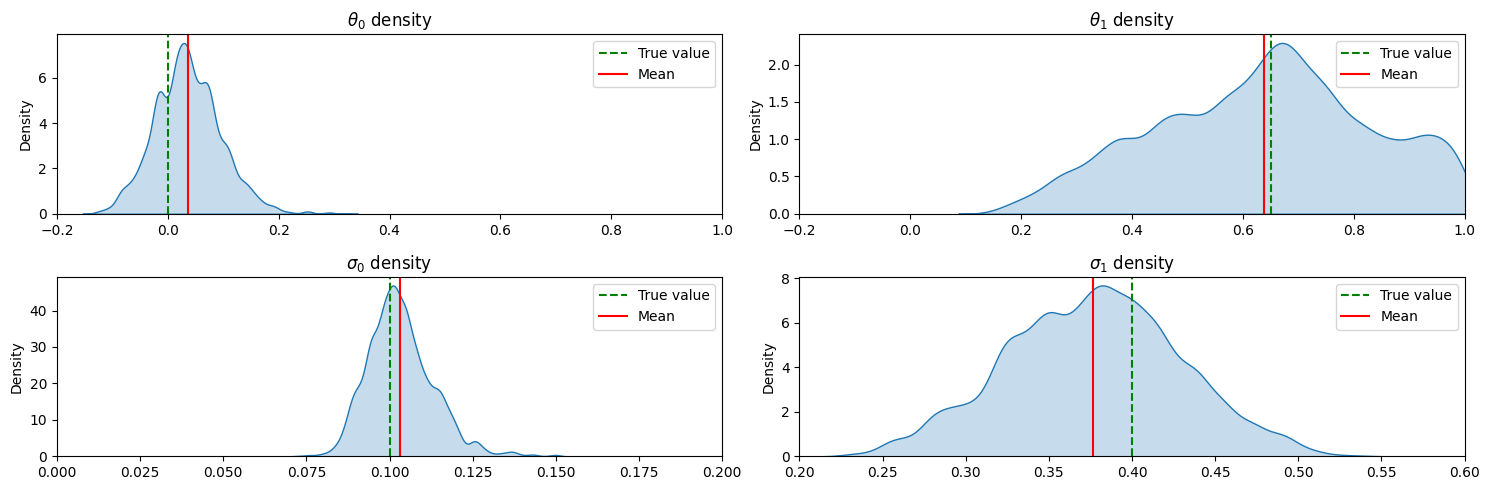}
	\caption{Densities after removing the first 2000 iterations which are considered burnin samples. \label{fig:densities} }
\end{center}
\end{figure}

\subsection{Application to shape analysis of butterfly wings}
We then proceed to a higher dimensional example where we model evolution of butterfly wing shapes perturbed by a Kunita-like flow. The use of Kunita flows in evolutionary models and application of BFFG for biological data is further explored in \cite{stroustrupStochasticPhylogeneticModels2025}, and the theoretical foundation for use of Kunita flows for shape stochastic is further detailed in \cite{bakerStochasticShapesKunitaFlows2025}. In evolutionary biology, trait evolution is commonly modelled with Brownian motions \cite{felsensteinPhylogeniesComparativeMethod1985}, traditionally for low dimensional data, i.e. small numbers of indiviual traits. Kunita flows provide a way of moving from this case to shape data that posses large numbers of points with high correlation between points. The Kunita flow model can thus be seen as a Brownian motion equivalent for morphological data. The idea is now given observations of the morphology at current time, i.e. at the leafs of the phylogenetic tree, to estimate properties of the evolutionary process at prior time, particularly parameters for the inter-point covariance and node values at prior time, e.g. the morphology at the start of the evolution at the root of the tree. To exemplify this, we consider the outline of butterfly wings represented by point configurations $x=(x_1, \ldots, x_n)$, $x_i \in \mathbb{R}^d$, $d=2$, and the Kunita flow
\begin{equation}
    dx_t^i
    =
    \int_{\mathbb{R}^d}k(x_t^i,\zeta)dW_t(\zeta)d\zeta
    \label{eq:kunita_flow_landmark}
\end{equation}
where $k$ is a kernel that describes the diffusivity operator of the Kunita flow in integral form. We let the forward model be a finite dimensional version of this with the SDE
\begin{equation}
    dx_t
    =
	K(x_t)dW_t
    \label{eq:landmark_sde}
\end{equation}
where $K$ is now the kernel matrix $[K(x)]_{ij} = k_{\bs{\th}}(x_i, x_j)\mathrm{Id}_d$, and $W_t$ is an $nd$-dimensional Wiener process. We let the kernel $k_{\bs{\th}}(x_i,x_j)=k_{\bs{\th}}(\|x_i-x_j\|)$ be a Laplacian kernel $k_{\bs{\th}}(r) = \alpha 4(3+3r+r^2)\exp(-r/\sigma)$. Conditioned on the leaf values, the shapes are observed with independent Gaussian noise with variance $\epsilon$ so that the parameter vector is $\bs{\th}=(\alpha,\sigma,\epsilon)$. 

{\it Backwards filtering: }For the auxiliary process \eqref{eq:auxiliary_sde}, we backwards filter on each  edge $e=(s,t)$  by taking $ B\equiv 0$ and $\beta\equiv 0$. On each edge, we redefine $\tilde\sigma_e$ iteratively. Let the choice of $\tilde\sigma_e$ on edge $e$ in iteration $i$  be denoted by $\tilde\sigma_e^{(i)}$, $i\ge 1$. Let $v(u)=H^{-1}(u) F(u)$ ($0\le u \le \tau_e$) and denote its value in iteration $i$ by $v^{(i)}(u)$. 
We set $\tilde\sigma_e^{(1)} = K\left(v^{(1)}(\tau_e)\right)$ and for $i\ge 1$
 \[\tilde\sigma^{(i+1)}(u)=\frac{u}{\tau_e}K\left(v^{(i+1)}(\tau_e)\right)+\left(1-\frac{u}{\tau_e}\right)K\left(v^{(i)}(0)\right).\] In this we, we aim to  iteratively refine the backwards filter.

The tree and observed leaf-node shapes are shown in Figure \ref{fig:butterfly_tree_and_shapes}. We then run an MCMC chain to estimate the parameters $\bs{\th}$ using Algorithm \ref{alg:mcmc} in Appendix \ref{app:alg}. The ``full conditional'' for updating the parameter $\epsilon$ is proportional to 
\[ (2\pi\eps)^{-|\scr{V}|/2} \exp\left(-\frac1{2\eps} \sum_{v\in \scr{V}} \|x_v-x_{\pa(v)}\|^2 \right) \pi(\epsilon)\]
Assuming the inverse-Gamma prior with parameter $(A_\eps, B_\eps)$ on $\epsilon$, it is easily seen that the full conditional distribution of $\epsilon$ is inverse-Gamma with parameter \[\left(A_\epsilon+\frac{|\scr{V}|}{2}, B_\epsilon +  \sum_{v\in \scr{V}} \|x_v-x_{\pa(v)}\|^2 \right).\] 
We took $A_\epsilon=2$ and $B_\epsilon=0.005$.
Trace plots for $\bs{\theta}$ are shown in Figure \ref{fig:butterfly_traces}. In Figure \ref{fig:butterfly_sample}, we show samples from the posterior root of the tree together with the original leaf values, and posterior samples of the parent of the leaves representing butterflies papilio xuthus and papilio zelicaon (two rightmost leaves of the tree).

\begin{figure}
\begin{center}
\includegraphics[width=0.48\textwidth]{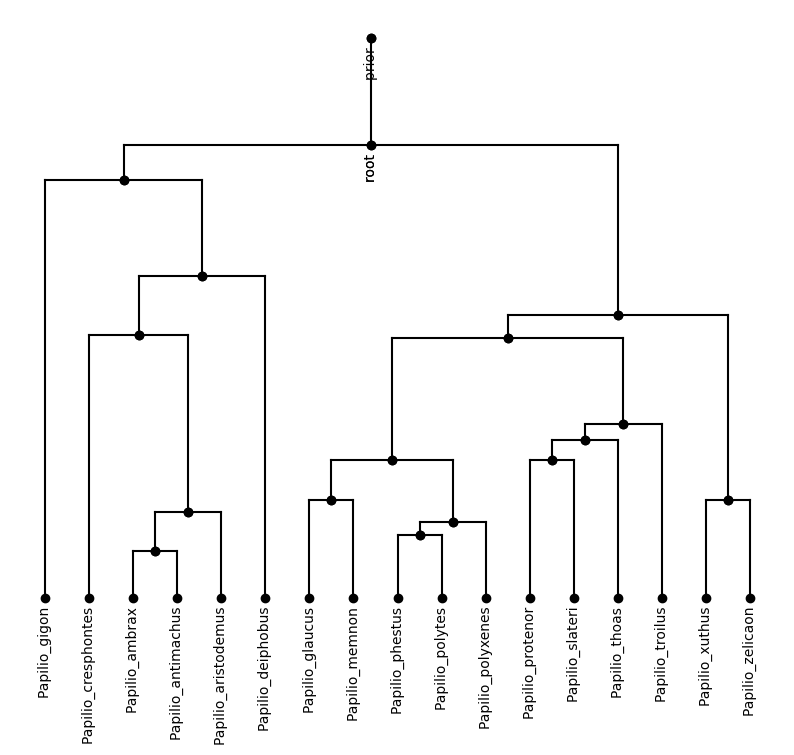}
\includegraphics[width=0.48\textwidth]{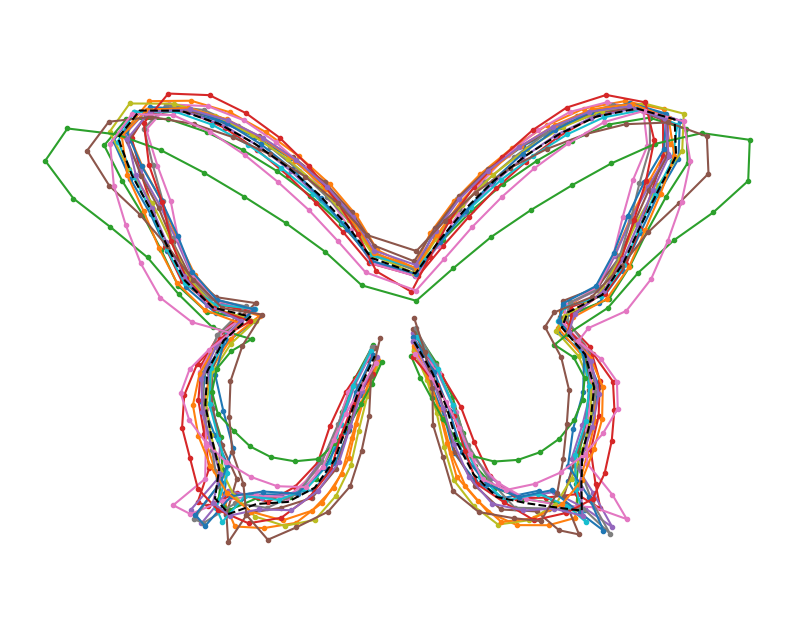}
\caption{Tree and observed shapes.} 
\label{fig:butterfly_tree_and_shapes}
\end{center}
\end{figure}

 \begin{figure}
\begin{center}
	\includegraphics[width=0.90\textwidth]{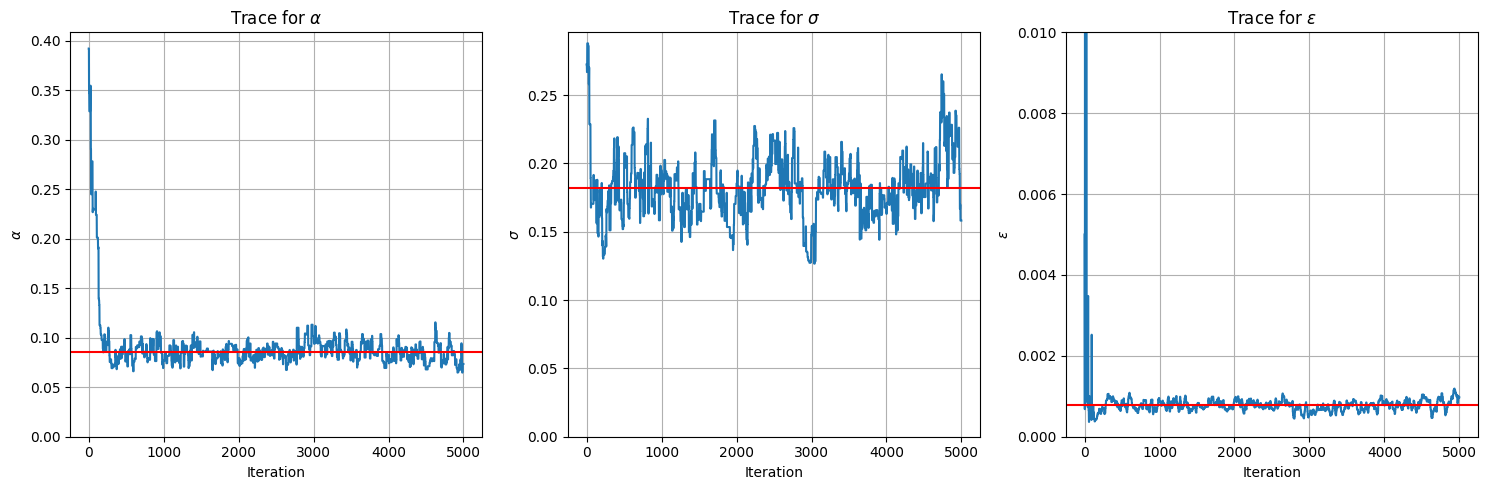}
	\includegraphics[width=0.90\textwidth]{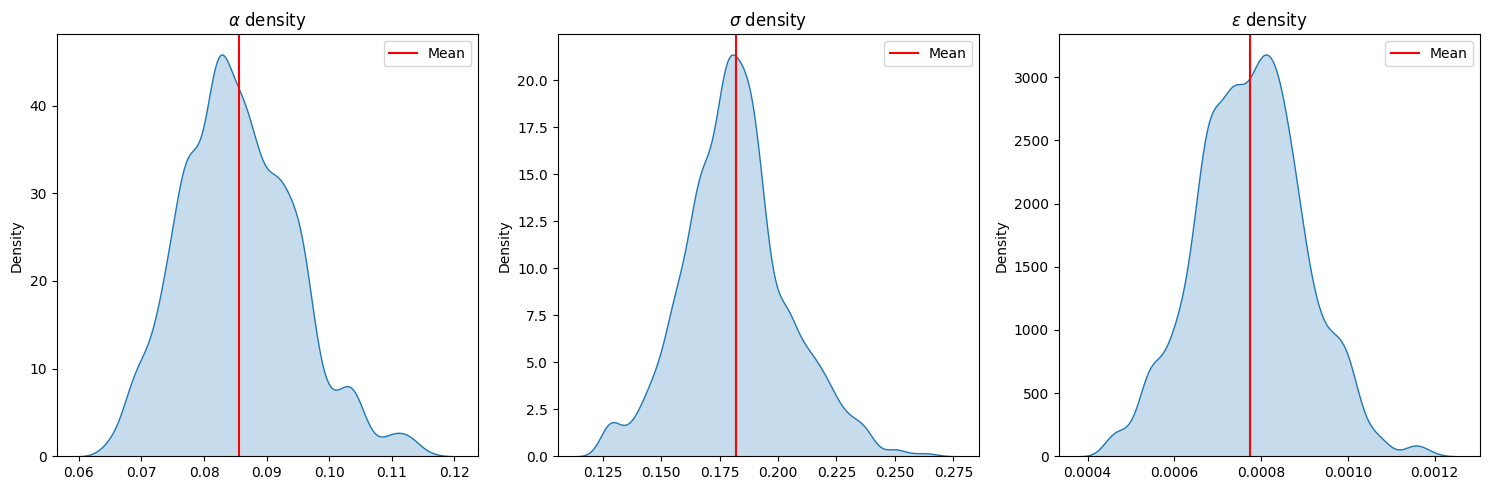}
	\caption{Trace and density plots for the parameters $\alpha$, $\sigma$ for the kernel used in the Kunita flow and the observation noise variance $\epsilon$. Red lines show the mean values of the samples.}
	\label{fig:butterfly_traces}
\end{center}
\end{figure}

 \begin{figure}
\begin{center}
\includegraphics[scale=0.4,trim=60 30 40 30,clip=true]{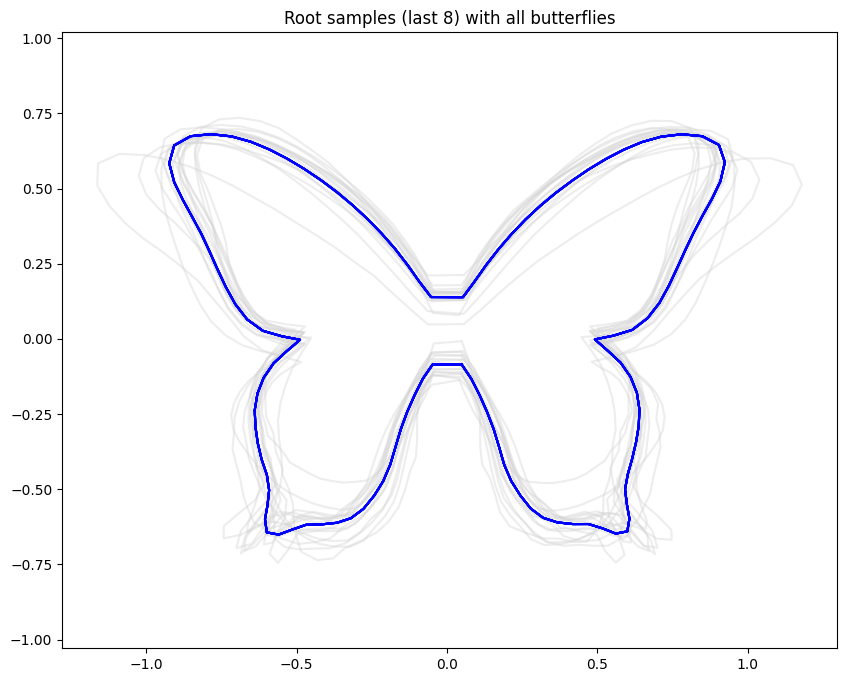}
\includegraphics[scale=0.4,trim=60 30 30 30,clip=true]{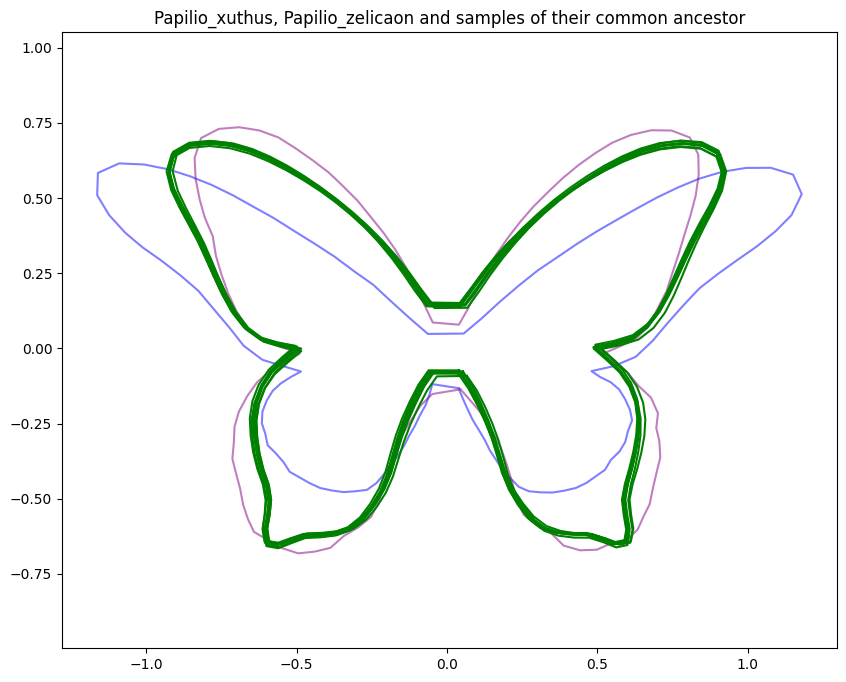}
\caption{(left) Posterior samples of the root of the tree (leaf values in the background). (right) Posterior samples of the parent of the leaves representing butterflies papilio xuthus and papilio zelicaon (continuous lines).}
\label{fig:butterfly_sample}
\end{center}
\end{figure}

\bigskip

\noindent {\bfseries Acknowledgement:}   

We thank Richard Kraaij for discussions on the topic of $h$-transforms in the early stage of this work and valuable feedback from Marc Corstanje.  The data for the butterfly example was kindly provided and preprocessed by Michael Baand Severinsen. We thank Frank Sch\"afer for preliminary simulation experiments in Julia language (\texttt{MitosisStochasticDiffEq.jl}-package) for the first  experiment in Section \ref{sec:sde_tree}.

\bibliographystyle{harry}
\bibliography{literature}

\appendix \label{sec:appendix}


\section{Auxiliary results for exponential change of measure}

\begin{theorem}[Corollary 3.3 of Chapter 2 in \cite{EthierKurtz1986}] \label{thm:loc_martingales}
Let $U$ and $Y$ be real-valued, right-continuous, $\cF_t$-adapted processes on the probability space $(\Omega, \scr{F}, \{\scr{F}_t,\, t\ge 0\}, \PP)$.
	Suppose that for each $t$, $\inf_{s\le t} U_s>0$. Then 
	\[ M_1(t)= U_t -\int_0^t Y_s \dd s \]
	is an $\cF_t$-local martingale if and only if 
\[ M_2(t)= U_t \exp\left(-\int_0^t \frac{Y_s}{U_s} \dd s\right)	\]
	is an $\cF_t$-local martingale	
\end{theorem}

For a function $f$ in the domain of the infinitesimal generator $\scr{L}$, we have that 
\[ M_t=f(X_t)-f(X_0)-\int_0^t (\scr{L} f)(X_s) \dd s \]
is an  $\cF_t$-martingale (Dynkin's martingale). We can extend this to the space-time process: for $f$ in the domain of $\cA = \cL + \partial_t$ the process 
\[ \bar{M}_t=f(t,X_t)-f(0,X_0)-\int_0^t (\scr{A} f)(s,X_s) \dd s \]
is an  $\cF_t$-martingale. 

\begin{corollary}\label{cor:sufficient_localmartingale} Suppose $X$ takes values in a  metric space $E$  and $g\colon [0,T] \times E \to \RR$. 
If $\inf_{s\le t} g(s,X_s)>0$, then 
\[ g(t,X_t) \exp\left(-\int_0^t \frac{(\cA g)(s,X_s)}{g(s,X_s)} \dd s\right) \]
is an $\cF_t$-local martingale 
\end{corollary}
\begin{proof}
Apply Theorem \ref{thm:loc_martingales} with   $U_s=g(s,X_s)$ and $Y_s =(\cA g)(s,X_s)$ 	.
\end{proof}

\begin{proposition}[Proposition 3.2 in \cite{PalmowskiRolski2002}]\label{prop_goodfunction}
 Suppose $g$ is a positive function that is in the domain of $\cA$. Then  $g$ is a good function if   either of the following conditions holds: \begin{itemize}
	\item $g\in \bB(E)$ and $(\cA g)/g \in \bB(E)$;
	\item $g, \scr{A}g \in \bB(E)$ and $\inf_{t,x}g_t(x)>0$. 
\end{itemize}
\end{proposition}

\section{Proof of Proposition \ref{prop: kl_backward}}\label{sec:remaining proofs}

We aim to find $g_i$ minimising $-\sum_{i=1}^d \int \pi(x) g(x)  \log g_i(x_i) \dd x$. Writing $\log g_i = \tilde g_i$ and introducing Lagrange-multiplier we consider
\[ \cL(\tilde g_1,\ldots, \tilde g_d, \lambda_1,\ldots, \lambda_d) = -\sum_{i=1}^d \pi(x) g(x) \tilde g_i(x_i) \dd x +\sum_{i=1}^d \lambda_i \left(\int e^{\tilde g_i(x_i)} \pi_i(x_i)-c_i\right). \]
Now consider $\cL$ with its $j$-th argument taken to be $\tilde g_j(x_j) + \epsilon \delta(x_j)$, subtract $\cL$ and collect terms of $\scr{O}(\epsilon)$ to get
\[ -\int \pi(x) g(x) \delta(x_j) \dd x + \lambda_j \int g_j(x_j) \pi_j(x_j) \delta(x_j) \dd x_j. \]
Equating this to zero gives
\[ \int \delta(x_j) \left( \lambda_j g_j(x_j) \pi_j(x_j)  -\int \pi(x) g(x) \dd x_{-j}\right) \dd x_j = 0\]
which implies 
\[  \lambda_j g_j(x_j) \pi_j(x_j)  =\int \pi(x) g(x) \dd x_{-j}. \]
Integrating over $x_j$ gives $\lambda_j c_j = 1$ from which the result easily follows.

\section{Algorithm for numerical results of Section \ref{sec:sde_tree}}\label{app:alg}

Assume a directed tree with inner edges that are all continuous edges. Denote the collection of all such edges by $\scr{E}$.  On edge $e$, assume the continuous-time process evolving on it is parametrised by the parameter $\theta$ and that there exists a map $\scr{F}_\theta$ and random process $Z^e$ such that $X^{\circ,e}=(X^\circ_u,\, u \in [0,\tau_e])=\scr{F}_\theta(x_0,Z^e)$, assuming $X^{\circ,e}_0=x_0$.  In the following, we assume the value of the root vertex, denoted by $x_r$ is known. Let $\scr{X}^\circ=(X^{\circ,e},\, e \in \scr{E})$, $\scr{Z} = (Z^e,\, e \in \scr{E})$ and write (with abuse of notation) $\scr{X}^\circ = \scr{F}_\theta(x_r, \scr{Z})$. 
Let
\begin{equation}
    \Psi(\scr{X}^\circ;\theta) =g_r(x_r;\theta) \exp\left( \sum_{e\in \mathcal{E}} \int_0^{\tau_e} \frac{\scr{A} g}{g}(u,X^\circ_u) \dd u    \right).
\end{equation}
Algorithm \ref{alg:mcmc} defines a Markov chain that has  the posterior distribution of $(\theta, \scr{Z})$ as invariant distribution. It is an MCMC algorithms that updates $\theta$ and $\scr{Z}$ from their conditional distributions. 

Let $\Pi$ denote the distribution of $\scr{Z}$ and $Q(\scr{Z}, \dd \scr{Z}^\circ)$ the Markov kernel that samples $\scr{Z}^\circ$ conditional on $\scr{Z}$. Define
\[ w(\scr{Z}, \scr{Z}^\circ) = \frac{\Pi(\dd \scr{Z}^\circ) Q(\scr{Z}^\circ, \dd \scr{Z})}{\Pi(\dd \scr{Z}) Q(\scr{Z}, \dd \scr{Z}^\circ)}.\]
Let $\overline{\Pi}$ denote the distribution of $\theta$ and $\overline{Q}(\theta, \dd \theta^\circ)$ the Markov kernel that samples $\theta^\circ$ conditional on $\theta$. Define
\[ \bar w(\theta, \theta^\circ) = \frac{\overline{\Pi}(\dd \theta^\circ) \overline{Q}(\theta^\circ, \dd \theta)}{\overline{\Pi}(\dd \theta) \overline{Q}(\theta, \dd \theta^\circ)}.\]
Finally, let $N$ denote the total number of iterations that the chain is simulated. \begin{algorithm}
\caption{MCMC-implementation of BFFG with parameter estimation.\label{alg:mcmc}}
\begin{algorithmic}[1]
\State \textbf{Initialize}: Sample $(\scr{Z},\theta) \sim \Pi \otimes  \overline{\Pi}$. Backward filter on the tree with $\theta$.   Set $\scr{X}=\scr{F}_\theta(x_r, \scr{Z})$. 
\For{$i = 0$ to $N$}
\State \textbf{Update path}
\State Sample  $\scr{Z}^\circ \sim Q(\scr{Z},\cdot)$. Set $\scr{X}^\circ=\scr{F}_\theta(x_r, \scr{Z}^\circ)$. Accept $\scr{Z}^\circ$ with probability $A\wedge 1$ where   \begin{equation*}
        A = \frac{\Psi(\scr{X}^\circ;\theta)}{\Psi(\scr{X};\theta)} w(\scr{Z}, \scr{Z}^\circ),
    \end{equation*}  
\State If accepted, set $\scr{Z}:=\scr{Z}^\circ$ and $\scr{X}:=\scr{X}^\circ$. 
\State \textbf{Update parameter $\theta$}
\State Sample $\theta^\circ$ from $\overline{Q}(\theta, \cdot)$.
\State Backward filter on the tree with $\theta^{\circ}$. 
\State Set $\scr{X}^\circ = \scr{F}_{\theta^\circ}(x_r, \scr{Z})$. 
\State Accept $\theta^\circ$ with probability $\bar A\wedge 1$, where 
    \[ \bar A = \frac{\Psi(\scr{X}^\circ;\theta^\circ)}{\Psi(\scr{X};\theta)} \bar w(\theta, \theta^\circ).\]
\State If accepted, set $\theta:=\theta^\circ$ and $\scr{X}:=\scr{X}^\circ$. 
\EndFor
\end{algorithmic}
\end{algorithm}

\subsection{Special case: stochastic differential equation}
If $X^{\circ,e}$ solves \eqref{eq:Xcirc_sde}, $Z^e$ is the driving Wiener process of the stochastic differential equation on edge $e$. Here, we assume a strong solution to the SDE exists.
In that case, one choice for $Q$ is the preconditioned Crank-Nicolson scheme, that proposes $Z^\circ_e$ conditional on $Z^e$ as follows
\begin{equation}\label{eq:pCN} Z^{\circ,e} := \lambda Z^e + \sqrt{1-\lambda^2} \bar{W}^e,\end{equation} where $W^e$ is a Wiener process on edge $e$ that is independent of $Z^e$ and  $\lambda \in [0,1)$ is a tuning parameter. For this choice $w(\scr{Z}, \scr{Z}^\circ)=1$.

\end{document}